%% file: ratings.tex
\newif\ifaap
\aapfalse

\ifaap
\documentclass[aap]{aptpub}
\usepackage{footmisc}
\else
\documentclass[11pt,a4paper]{article}
\fi

\ifaap
\shorttitle{A Markov approach to credit rating migration conditional on economic
  states}
\authornames{M.\ Kalkbrener and N.\ Packham}
\fi

\ifaap
\providecommand{\citet}{\cite}
\providecommand{\citep}{\cite}
\else
\RequirePackage{amsthm,amsmath,amsfonts,amssymb}
\usepackage[longnamesfirst]{natbib}
\bibpunct{(}{)}{;}{a}{,}{,}
\fi
\usepackage[normalem]{ulem} %
\usepackage{enumerate}
\ifaap\else
\usepackage{fullpage}
\fi

\ifaap\else
\usepackage{hyperref}
\fi
\usepackage{graphicx,color}
\graphicspath{{./_pics/}}
\definecolor{BrickRed}{rgb}{.625,.25,.25}

\definecolor{markergreen}{rgb}{0.6, 1.0, 0}
\definecolor{darkgreen}{rgb}{0, .5, 0}
\definecolor{darkred}{rgb}{.7,0,0}

\ifaap\else
\theoremstyle{plain}
\newtheorem{theorem}{Theorem}[section]
\newtheorem{proposition}[theorem]{Proposition}
\newtheorem{corollary}[theorem]{Corollary} 
\newtheorem{lemma}[theorem]{Lemma} 
\fi
 \ifaap
 \else
\theoremstyle{definition} 
\newtheorem{definition}[theorem]{Definition}

\newtheorem{assumption}{Assumption}
 \fi

\usepackage{bm}

\usepackage{tikz}
\usetikzlibrary{positioning}

\input{definitions}

\begin{document}

\title{
  \ifaap\else\Large\bf\fi
  A Markov approach to credit rating migration\\ conditional on economic
  states 
} %

\ifaap
\else %
\renewcommand{\thefootnote}{\fnsymbol{footnote}} %
\author{Michael Kalkbrener\footnotemark,\ \  Natalie Packham\footnotemark} %
\date{\today\medskip
}
\maketitle %
\fi

\begin{abstract}
  We develop a model for credit rating migration that accounts for the 
  impact of economic state fluctuations on default probabilities. 
  The joint process for the economic state and the rating is modelled
  as a time-homogeneous Markov chain. While the rating process itself
  possesses the Markov property only under restrictive conditions,
  methods from Markov theory can be used to derive the
  rating process' asymptotic behaviour. We use the mathematical
  framework to formalise and analyse different rating philosophies,
  such as point-in-time (PIT) and through-the-cycle (TTC) ratings.
  Furthermore, we introduce stochastic orders on the bivariate process'
  transition matrix to establish a consistent notion of ``better'' and
  ``worse'' ratings. Finally, the construction of PIT and TTC ratings
  is illustrated on a Merton-type firm-value process.
\end{abstract}

\ifaap
\else
 \medskip

\footnotetext[1]{%
  Michael Kalkbrener, Deutsche Bank AG, Otto-Suhr-Allee 16,
 10585 Berlin, Germany.}

\footnotetext[2]{  
  Natalie Packham, Berlin School of Economics and Law,
  Badensche Str.\ 52, 10825 Berlin, Germany. 

\ \ Email:
  packham@hwr-berlin.de \smallskip
  
\noindent
  The views expressed in this paper are those of the authors
  and do not necessarily reflect the position of Deutsche Bank
  AG.\medskip} %
\fi
\renewcommand{\thefootnote}{\arabic{footnote}} %


\section{Introduction}
\label{sec:introduction}

It is
generally accepted that corporate default rates depend on 
economic conditions 
\citep{Wilson1998,Nickell2000,Bangia2002,Heitfield2005,Malik2012}. Rating
systems, which are designed to provide information on the credit 
quality of corporates, differ with respect to the type of economic
information reflected in the rating. The two main rating philosophies
are {\em through-the-cycle (TTC) ratings}, which are insensitive to
economic conditions and focus on company-specific criteria, and {\em
  point-in-time (PIT) ratings}, which reflect the current economic
situation as well as expectations of future economic developments. 
However, no precise definition of PIT and TTC characteristics
of rating systems has been established. One objective of this paper is
to develop a mathematical 
framework for rating migration processes, factoring in the
economic situation, that allows for a formal characterisation
of PIT and TTC ratings. 

The distinction between PIT and TTC rating concepts is particularly
relevant for banks who need to comply with accounting standards and capital regulation.
Assessing the credit risk of a company solely by its probability
  of default (PD), incentivises procyclical adjustments of financial
  institutions' capital requirements, i.e., decreasing default
rates in economic upturns would allow financial institutions to reduce
capital for their corporate credit portfolios, while increasing default rates in economic downturns
would increase capital requirements, see e.g.\
\citet{Borio2001,Behn2016} and references therein. As this procyclical
capitalisation poses a threat to financial stability, financial
regulators prefer banks to assess regulatory capital requirements for
credit risk using TTC ratings and default probabilities,\footnote{Capital requirements
  regulation and directive -- CRR/CRD IV, Article 502} which are
designed to be stable through the business cycle. In contrast, 
the International Financial Reporting Standard (IFRS) 9 as well as the
accounting standard Current Expected Credit 
Losses (CECL) require financial institutions to calculate expected
credit losses that reflect current economic conditions and forecasts
of future economic conditions,\footnote{International Financial
  Reporting Standard (IFRS) 9, Paragraph 5.5.17.
  For CECL, see FASB Accounting Standards Update (ASU)
  2016-13.} leading to
the challenging problem of incorporating the impact of economic
factors into the estimation of default probabilities. 

Models of rating migration processes specified as
time-homogeneous Markov chains, as is common in the credit risk
literature \citep{Duffie2003,Jarrow1997,Lando2004,Bluhm2003}, lack the
necessary flexibility to capture the complex link to the macroeconomic
environment. To resolve this problem, several extensions of the simple
Markov framework have been proposed in the literature: 
regime-switching models specify different rating migration matrices
for periods of economic contraction and expansion \citep{Bangia2002},
mixtures of Markov chain models define rating migration processes as a
combination of several Markov chains \citep{Frydman2008,Fei2012}.

The model presented in this paper extends the above-mentioned
approaches by incorporating an economic state process with finitely
many states. Our starting point is to model the evolution of economic
states $\mathcal A$ as a time-homogeneous Markov chain $(A_t)_{t\in 
  \Nzero}$. Rating migrations $(R_t)_{t\in \Nzero}$, 
specified on the state space of rating classes $\mathcal R$, are not
assumed to follow a Markov process but are defined by transition
matrices that depend on the state of the economy. However, it is a key
assumption that the rating process has the Markov
property if its state space is extended to the product space $\mathcal
A\times \mathcal R$ consisting of economic states and rating classes. 
More formally, the combined process $(X_t)_{t\in
  \Nzero}=(A_t,R_t)_{t\in \Nzero}$ defined on the product space
$\mathcal A\times \mathcal R$ of economic states and ratings is
assumed to be a time-homogeneous Markov chain.

In the first part of the paper (Sections
\ref{sec:defin-cond-rating}--\ref{sec:asympt-prop}), we study
conditions on the conditional 
transition matrices that guarantee the Markov property for $(R_t)_{t\in
  \Nzero}$ within the state space $\mathcal R$ itself. Since
$(R_t)_{t\in \Nzero}$ is a projection of the time-homogeneous Markov
chain $(X_t)_{t\in \Nzero}$, this problem is a special case of the more
general question of whether functions of Markov chains possess the Markov
property \citep{Rosenblatt1971}. We show that only a small
subclass of TTC rating migration processes are (time-inhomogeneous)
Markov chains on the state space $\mathcal R$, which provides
theoretical support to the claim that real-world rating systems do not
have the Markov property, e.g. \citet{Altman1998,Lando2002,Guettler2010}.
However, it can be 
shown that for any non-Markovian rating process $(R_t)_{t\in \Nzero}$
there exists a time-inhomogeneous Markov chain $({\widetilde
  R}_t)_{t\in \Nzero}$on $\mathcal R$ with time-dependent transition
matrices $({\widetilde P}_t)_{t\in {\mathbb N}}$ that has identical
default rates and future rating distributions as $(R_t)_{t\in
  \Nzero}$, i.e., \[  \p({\widetilde R}_t=r) = \p(R_t=r), \quad t\in
  \Nzero, \quad r\in{\mathcal R}, \] where $\p(R_t=r)$ denotes the
probability of $R_t$ being in rating class $r$. 
Using the Perron-Frobenius Theorem, we prove that the sequence of
transition matrices $({\widetilde P}_t)_{t\in {\mathbb N}}$ converges
to a limit, denoted by ${\widetilde P}_{\infty}$. The corresponding
time-homogeneous Markov chain on $\mathcal R$ with transition matrix
${\widetilde P}_{\infty}$ replicates the rating distribution and
default behaviour of $(R_t)_{t\in \Nzero}$ in the asymptotic limit. In
particular, ${\widetilde P}_{\infty}$ specifies the vector of
asymptotic default probabilities of the different rating classes. 

The explicit incorporation of an economic state process into the
specification of the rating process $(R_t)_{t\in \Nzero}$ allows for a
unified analysis of different rating philosophies. In the second part
of the paper (Section \ref{sec:rating-philosophies}), we provide a
formal definition of PIT and TTC rating 
migration processes based on a decomposition of conditional transition
matrices into default and non-default components. In a nutshell, TTC
ratings have a firm-specific character, i.e., the non-default rating
transitions are insensitive to economic state fluctuations, whereas
PIT non-default
ratings are explicitly assigned according to a firm's PD (probability of
default), whence they swing with economic state fluctuations. 

The third part of the paper (Section \ref{sec:ordered-ratings}) is 
concerned with orderings on distributions on the state space $\mathcal
A\times \mathcal R$ of the combined process $(X_t)_{t\in \Nzero}$,
formalising the 
notions of ``better'' or ``worse'' economic states and rating
states. A natural approach to order rating states is by
first-order stochastic dominance, see \citet{Jarrow1997}, which,
amongst other properties, preserves the order of long-term default
probabilities. Several extensions of first-order stochastic
dominance exist in a multivariate setting, amongst which stochastic
dominance and the upper orthant order turn out to be relevant for the
combined process $(X_t)_{t\in \Nzero}$. We analyse conditions for stochastic
monotonicity of transition matrices of $(X_t)_{t\in \Nzero}$
as well as conditions for preserving the order of multi-year default
probabilities conditional on the initial economic state and rating.

In the fourth part (Section \ref{sec:stylised-example}), we
apply the theoretical framework to create PIT and TTC ratings using a
Merton-type firm-value process \citep{Merton1974}. The economic state
influences the firm-value by adding a systematic drift. PIT ratings
are determined by classifying firms according to their PDs. While the
PIT rating construction is relatively straightforward, there are
various methods to build TTC ratings, such as considering a firm's
idiosyncratic risk or stressed PDs. In a worked example we select the
TTC rating that most closely aligns with the PIT rating by minimising
the differences in multi-year PDs conditional on various combinations
of initial economic states and PIT ratings. This provides valuable
insights into how PIT and TTC ratings coexist, as well as their
similarities and differences.

We conclude by demonstrating how rating models within our framework  
concurrently meet requirements of ratings and PDs in both accounting
and regulatory capital standards (Section \ref{sec:conclusion}). 

\section{Conditional rating transition matrices}
\label{sec:defin-cond-rating}

\subsection{Definition}
\label{sec:definition}

We refer to \citet{Norris1998} for basic definitions in the theory of 
Markov chains.
Let $(\Omega, \mathcal F, \p)$ be a probability space. On this
probability space there exists a stochastic process $(A_t)_{t\in \Nzero}$
taking values in a set 
$\mathcal A$, which consists of $K$ economic states. The economic
state process is assumed to be a time-homogeneous Markov chain with
transition matrix $M=(m_{ab})\in [0,1]^{K\times K}$, i.e.,
$\p(A_{t+1}=b|A_t=a)=m_{ab}$ for $t\in \Nzero$.

In addition, a stochastic process $(R_t)_{t\in \Nzero}$ denoting the
rating process of a firm exists, taking values in the set
$\mathcal R$. The set $\mathcal R$ has $J$ elements, which are
interpreted as rating states including an absorbing default state
$\overline{r}$.  
The subset of non-default ratings is denoted by $\widehat{\mathcal
  R}$, i.e., $\widehat{\mathcal R}={\mathcal
  R}\setminus\{\overline{r}\}$.  

We do not require the rating process $(R_t)_{t\in \Nzero}$ to have the
Markov property. Instead we make the weaker assumption that the joint
process $(X_t)_{t\in \Nzero}=(A_t,R_t)_{t\in \Nzero}$ taking values in
the product space $\mathcal X=\mathcal A\times \mathcal R$ is a
time-homogeneous Markov chain with transition matrix
\begin{equation}
P=(p_{(a,r),(b,s)})\in [0,1]^{K J\times K J},
\label{eq:4}
\end{equation}
where
\begin{equation*}
  p_{(a,r),(b,s)} =
  \p(A_{t+1}=b,R_{t+1}=s|A_t=a,R_t=r) 
\end{equation*}
denotes the probability of moving from the state $(a,r)$ to
$(b,s)$. The distribution of $X_0$ is denoted by $\lambda\in [0,1]^{K
  J}$. W.l.o.g.\ we assume a default probability of 0 at time 0, i.e.,
$\lambda_{(a,\overline{r})}=0$ for $a\in{\mathcal A}$.

It is a reasonable assumption that if the current economic state is
known, knowledge about the current rating does not add information
about the future economic state. The following Lemma establishes that,
under this assumption, the transition matrix $P$ is fully determined by 
\begin{enumerate}
\item the transition probabilities $m_{ab}$ of the economic state
  process and  
\item the rating transition probabilities conditional on economic
  state transitions, i.e., the conditional probabilities 
$\p(R_{t+1}=s|R_t=r, A_{t+1}=b, A_t=a)$. 
\end{enumerate}
To formalize this statement, conditional rating migration matrices 
$M^{(a,b)} = (m_{rs}^{(a,b)})\in [0,1]^{J\times J}$ are defined for
each economic state transition $(a,b)\in{\mathcal A}\times{\mathcal
  A}$, where 
\begin{equation*}
  m_{rs}^{(a,b)} = \p(R_{t+1}=s|R_t=r, A_{t+1}=b, A_t=a). 
\end{equation*}
We assume that migrations from non-default ratings to any other rating
class have positive probability, i.e., $m_{rs}^{(a,b)}>0$ for all
$r\in\widehat{\mathcal R}$, $s\in {\mathcal R}$ and $a,b\in \mathcal
A$. Since $\overline{r}$ is absorbing, $m_{\overline{r} r}^{(a,b)}=0$
and $m_{\overline{r}\, \overline{r}}^{(a,b)}=1$.

\begin{lemma}~\label{200125:l1}
$p_{(a,r),(b,s)}=m_{ab} m_{rs}^{(a,b)}$ if and only if $A_{t+1}$ and
$R_t$ are conditionally independent given $A_t$. 
\end{lemma}
\begin{proof}
  ``$\Leftarrow$'': Using that
  $\p(A_{t+1}=b|A_t=a)=\p(A_{t+1}=b|A_t=a,R_t=r)$:
  \begin{align*}
    m_{ab} m_{rs}^{(a,b)} &= \p(A_{t+1}=b|A_t=a)\,
                            \p(R_{t+1}=s|R_t=r,A_{t+1}=b,A_t=a) \\
                          &= \p(A_{t+1}=b|A_t=a, R_t=r)\,
                            \p(R_{t+1}=s|R_t=r,A_{t+1}=b,A_t=a) \\
                          &= \p(A_{t+1}=b,R_{t+1}=s|R_t=r,A_t=a)\\
                          &= p_{(a,r),(b,s)}. 
  \end{align*} 
  ``$\Rightarrow$'' follows in the same way.
\end{proof}

From now on we assume that $A_{t+1}$ and $R_t$ are conditionally
independent given $A_t$. It immediately follows from $P=(m_{ab}
m_{rs}^{(a,b)})$ that the process $X_t$ as well as the rating process
$R_t$, which can be considered as the projection of $X_t$ onto the
state space ${\mathcal R}$, are completely specified by the initial
distribution $\lambda\in [0,1]^{K J}$ of $(A_0,R_0)$, the transition
matrix $M\in [0,1]^{K \times K}$ of $A_t$ and the $K^{2}$ conditional
rating migration matrices $M^{(a,b)}\in [0,1]^{J\times J}$. The
probability of a path of economic states and ratings has the form
\begin{equation*}
  \p(X_0=(a_0,r_0),\ldots,X_T=(a_T,r_T))=\p(X_0=(a_0,r_0))\prod_{t=0}^{T-1}
  m_{a_t a_{t+1}} m_{r_t r_{t+1}}^{(a_t,a_{t+1})}. 
\end{equation*}

\subsection{Example of a rating system}
\label{sec:example-rating-syst}

The following example shows the economic state transition matrix $M$ and the $M^{(a,b)}$ and $P$ matrices of a
rating system. It is derived from a Merton model (\cite{Merton1974}),
as will be explained in detail in Section \ref{sec:stylised-example}.

Economic state transitions are assumed to follow a Markov process with
transition matrix 
\begin{equation}
M = \begin{pmatrix}
  0.8 & 0.175 & 0.025\\
  0.1 & 0.8 & 0.1\\
  0.025 & 0.175 & 0.8
\end{pmatrix}. 
\end{equation}
Here, the first state is considered to be a good, the
second state a neutral and the third state a
bad economic state. 

\begin{figure}[t]
  \centering
  \includegraphics[width=.45\textwidth]{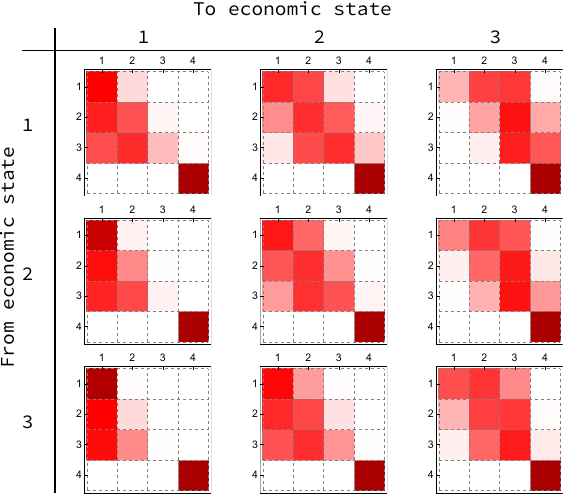}\ \ \
  \ \ \
  \includegraphics[width=0.4\textwidth]{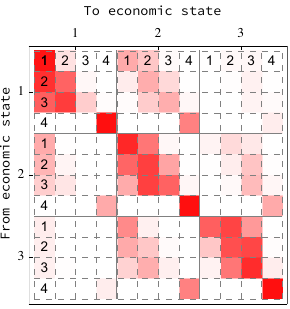}\ \ \
  \caption{$M^{(a,b)}$ matrices (left) and $P$ matrix
    (right). Economic states are labelled 1 (good), 2 (neutral) and 3 
    (bad). Ratings are ordered from the highest rating class to the
    worst (default). Darker shades indicate values closer to one,
    whereas white indicates a value of zero or close to zero.}
  \label{fig:rating_example}
\end{figure}

There are four rating classes, of which the first three -- ordered
from best to worst -- are non-default ratings and the fourth is
the (absorbing) default state. The left hand side of Figure
\ref{fig:rating_example} shows the $M^{(a,b)}$ matrices for economic
states $a,b\in \{1,2,3\}$. In this example, the economic state
transition $(a,b)$ has a strong impact on rating migration: if $b$ is
the best economic state 1, upgrades clearly dominate downgrades (as
illustrated by the three matrices in the first column), whereas a
transition to the worst state $b=3$ increases the probability of
downgrades. Two examples of $M^{(a,b)}$ matrices are: 
\begin{equation*}
  \label{eq:11}
  M^{(2,1)} =
  \begin{pmatrix}
    0.9808& 0.0192& 0.0001& 4\cdot 10^{-7}\\
    0.8371& 0.1598& 0.0030& 0.0001\\
    0.6137& 0.3660& 0.0195& 0.0008\\
    0& 0& 0& 1
  \end{pmatrix},\ \ \
  M^{(2,3)} = 
  \begin{pmatrix}
    0.1737 & 0.5058 & 0.3188 & 0.0017 \\
    0.0213 & 0.2460 & 0.7004 & 0.0323 \\
    0.0033 & 0.0911 & 0.7814 & 0.1242\\
    0 & 0 & 0 & 1
  \end{pmatrix}.
\end{equation*}

The right-hand side of Figure \ref{fig:rating_example} shows the
$P$ matrix resulting from $p_{(a,r),(b,s)}=m_{ab} m^{(a,b)}_{rs}$, for
all $a,b\in \mathcal A$ and $r,s\in \mathcal R$. 

\section{Analysis of the Markov property of rating processes}
\label{sec:markovian}

\subsection{A class of Markovian rating processes}

As a consequence of the dependence  on the underlying economic
states, rating processes $(R_t)_{t\in \Nzero}$ typically do not have
the Markov property if their state space is restricted to the set of rating
classes. In this section, we establish conditions on the transition
matrices $M^{(a,b)}$ that ensure that $(R_t)_{t\in \Nzero}$ is a
Markov process on $\mathcal R$.  

It is a key model assumption in this paper that rating processes have
the Markov property on the extended state space $\mathcal X=\mathcal
A\times \mathcal R$, which consists of all combinations of economic
states and rating classes. More formally, $(R_t)_{t\in \Nzero}$ is the
projection of the time-homogeneous Markov chain $(X_t)_{t\in \Nzero}$
with state space $\mathcal X$ onto the state space $\mathcal
R$. Hence, the question whether $(R_t)_{t\in \Nzero}$ is a Markov
process can be addressed in the more general context of studying the
Markov property of functions of Markov processes, see
\citet{Rosenblatt1971}.

The general results in \citet{Rosenblatt1971} can be significantly
improved in our setup by utilizing the structure of the transition
matrix $P=(p_{(a,r),(b,s)})\in [0,1]^{K J\times K J}$ given in Lemma
\ref{200125:l1}, i.e., representing $p_{(a,r),(b,s)}$ in the form
$p_{(a,r),(b,s)}=m_{ab} m_{rs}^{(a,b)}$. 
Not surprisingly, severe restrictions on the conditional transition
matrices $M^{(a,b)}$ are required to ensure the Markov property of the
rating process $(R_t)_{t\in \Nzero}$. It turns out that the following
definition formalizes the key property of conditional transition
matrices $M^{(a,b)}$ of Markovian rating processes $(R_t)_{t\in
  \Nzero}$.  

\begin{definition}
The conditional transition matrices $M^{(a,b)}$ 
have {\em identical ratios of non-default migration probabilities} if 
for all $r,r',s,s'\in \widehat{\mathcal R}$ and $a,a',b,b'\in \mathcal A$
  \begin{equation}
\frac{m_{rs}^{(a,b)}}{m_{r's'}^{(a,b)}}=\frac{m_{rs}^{(a',b')}}{m_{r's'}^{(a',b')}}.
    \label{200126:1}
  \end{equation}
\end{definition}
It is easily verified that the two matrices provided in the example in Section
\ref{sec:example-rating-syst} do not fulfil this:
$m_{1,1}^{(2,1)}/m_{1,2}^{(2,1)} > 1$, while
  $m_{1,1}^{(2,3)}/m_{1,2}^{(2,3)}<1$. 

As an example of a simple class of rating migration processes that
have identical ratios of non-default migration probabilities, consider
processes based on just two ratings $r$ and $\overline{r}$. Since $r$
is the only element of $\widehat{\mathcal R}$ condition
(\ref{200126:1}) is always satisfied, regardless of the specification
of the conditional transition matrices. 

The above condition significantly reduces the free parameters that are
required for specifying conditional transition matrices. Choose
$a,b\in {\mathcal A}$ and $r,s\in \widehat{\mathcal R}$ and specify
$M^{(a,b)}$ and $m_{rs}^{(a',b')}$ for all $(a',b')\not= (a,b)$. All
other parameters $m_{r's'}^{(a',b')}$ are defined by (\ref{200126:1}),
i.e., for $(r',s')\not=(r,s)$ 
\[ m_{r's'}^{(a',b')}=
\frac{m_{r's'}^{(a,b)}}{m_{rs}^{(a,b)}}
m_{rs}^{(a',b')}.
\]
Since $M^{(a,b)}$ is a stochastic matrix with one absorbing state, only 
$(J-1)^2+K^2-1$ free parameters are required for fully specifying
conditional transition matrices with identical ratios of non-default
migration probabilities. In the case $K=1$, just one rating transition
matrix is defined, which is equivalent to the standard specification
of a time-homogeneous Markov chain. For $K>1$, the number of
conditional rating migration matrices $M^{(a,b)}$ as well as the
number of free parameters increases by $K^2$.  

The set of Equations (\ref{200126:1}) is a sufficient condition for
the Markov property, see Corollary \ref{200322:c1}. In fact, Theorem
\ref{200126:t1} shows that the Equations (\ref{200126:1}) are
equivalent to the economic state variable $A_t$ being independent of
the rating history up to time $t$ (conditional on no default), which
is a slightly stronger condition than the Markov property of
$(R_t)_{t\in \Nzero}$. 

More formally, let $t\in\Nzero$. 
Conditional on no default, the economic state variable $A_t$ is
independent of the rating history up to time $t$ if for all $a\in
\mathcal A$ and $r_0,s_0,\ldots,r_t,s_t\in \widehat{\mathcal R}$ with
$\p(R_0=r_0,\ldots,R_t=r_t)>0$ and $\p(R_0=s_0,\ldots,R_t=s_t)>0$,   
 \begin{equation}
\p(A_{t}=a|R_0=r_0,\ldots,R_t=r_t)
=\p(A_{t}=a|R_0=s_0,\ldots,R_t=s_t)= \p(A_t = a| R_t\not=\overline r).
    \label{200123:3}
  \end{equation}

\begin{theorem} ~\label{200126:t1}
The following conditions are equivalent.
 \begin{enumerate}[(i)]
  \item The conditional transition matrices $M^{(a,b)}$ have identical
    ratios of non-default migration probabilities. 
  \item For every joint distribution $\lambda\in [0,1]^{K J}$ of
    independent variables $A_0$ and $R_0$ and every stochastic matrix
    $M\in [0,1]^{K \times K}$ the rating process $R_t$ specified by
    the distribution $\lambda$ and transition matrices $M$ and
    $M^{(a,b)}$ satisfies (\ref{200123:3}) for every $t\in \Nzero$. 
  \end{enumerate}
\end{theorem}
\begin{proof}
  In this proof we abbreviate a rating path
  $\{R_0=r_0,\ldots,R_t=r_t\}$ and a path 
  $\{A_0=a_0,\ldots,A_t=a_t\}$ of economic states by $r_{0..t}$ and
  $a_{0..t}$ respectively.
  
``$\Rightarrow$'': 
For $t\in \Nzero$ let $a_0,b_0,\ldots,a_t,b_t\in
\mathcal A$, $r_0,s_0,\ldots,r_t,s_t\in \widehat{\mathcal R}$ with
$\p(a_{0..t})>0$ and $\p(b_{0..t})>0$. We assume that  $A_0$ and $R_0$
are independent variables and $\p(R_0=r_0)>0$ and $\p(R_0=s_0)>0$. %
We first prove by induction on $t$ that
\begin{equation}
\frac{\p(r_{0..t}|a_{0..t})}{\p(s_{0..t}|a_{0..t})}=
\frac{\p(r_{0..t}|b_{0..t})}{\p(s_{0..t}|b_{0..t})}.
\label{231117:1} 
\end{equation}
Since $A_0$ and $R_0$ are independent, 
Equation~(\ref{231117:1}) holds for $t=0$. For $t>0$, we assume that
(\ref{231117:1}) holds for $t-1$. Note that 
\begin{equation}
\p(r_{0..t}|a_{0..t})=\frac{\p(r_{0..t-1},a_{0..t-1})m_{a_{t-1}a_t}m_{r_{t-1}r_t}^{(a_{t-1},a_t)}}
{\p(a_{0..t-1})m_{a_{t-1}a_t}}=\p(r_{0..t-1}|a_{0..t-1})m_{r_{t-1}r_t}^{(a_{t-1},a_t)}. \label{231219:1}  
\end{equation}
Hence, by induction and (\ref{200126:1}),
\[
\frac{\p(r_{0..t}|a_{0..t})}{\p(s_{0..t}|a_{0..t})}=
\frac{\p(r_{0..t-1}|a_{0..t-1})m_{r_{t-1}r_t}^{(a_{t-1},a_t)}}{\p(s_{0..t-1}|a_{0..t-1})m_{s_{t-1}s_t}^{(a_{t-1},a_t)}}= \\
\frac{\p(r_{0..t-1}|b_{0..t-1})m_{r_{t-1}r_t}^{(b_{t-1},b_t)}}{\p(s_{0..t-1}|b_{0..t-1})m_{s_{t-1}s_t}^{(b_{t-1},b_t)}}=
\frac{\p(r_{0..t}|b_{0..t})}{\p(s_{0..t}|b_{0..t})}.
\]

For proving (\ref{200123:3}) we apply Bayes formula. For $t=0$,
Equation (\ref{200123:3}) is implied by the independence of $A_0$ and
$R_0$. Otherwise it follows from (\ref{231219:1}) that Bayes formula
can be expressed as 
\begin{equation}
\p(a_{0..t}|r_{0..t})=\frac{\p(a_{0..t})\p(r_{0..t-1}|a_{0..t-1})m_{r_{t-1}r_t}^{(a_{t-1},a_t)}}
{\sum_{(b_0,\ldots,b_t)\in{\mathcal A}^{t+1}}\p(b_{0..t})\p(r_{0..t-1}|b_{0..t-1})m_{r_{t-1}r_t}^{(b_{t-1},b_t)}}.
\label{231117:2} 
\end{equation}
Hence, $\p(a_{0..t}|r_{0..t})=\p(a_{0..t}|s_{0..t})$ is equivalent to
\begin{multline}
\sum_{(b_0,\ldots,b_t)\in{\mathcal A}^{t+1}} \p(b_{0..t})
\p(r_{0..t-1}|a_{0..t-1})\p(s_{0..t-1}|b_{0..t-1})
m_{r_{t-1}r_t}^{(a_{t-1},a_t)}m_{s_{t-1}s_t}^{(b_{t-1},b_t)}=
\\
\sum_{(b_0,\ldots,b_t)\in{\mathcal A}^{t+1}} \p(b_{0..t})
\p(r_{0..t-1}|b_{0..t-1})\p(s_{0..t-1}|a_{0..t-1})
m_{s_{t-1}s_t}^{(a_{t-1},a_t)}m_{r_{t-1}r_t}^{(b_{t-1},b_t)}.
\label{231121:1}
\end{multline}
By (\ref{231117:1}), Equation (\ref{231121:1}) holds if and only if
\[ \sum_{(b_0,\ldots,b_t)\in{\mathcal A}^{t+1}} \p(b_{0..t})
\p(r_{0..t-1}|a_{0..t-1})\p(s_{0..t-1}|b_{0..t-1})
(m_{r_{t-1}r_t}^{(a_{t-1},a_t)}m_{s_{t-1}s_t}^{(b_{t-1},b_t)}-
m_{s_{t-1}s_t}^{(a_{t-1},a_t)}m_{r_{t-1}r_t}^{(b_{t-1},b_t)})=0.
\]
Since the conditional transition matrices $M^{(a,b)}$ have identical
ratios of non-default migration probabilities it follows that
$m_{r_{t-1}r_t}^{(a_{t-1},a_t)}m_{s_{t-1}s_t}^{(b_{t-1},b_t)}- 
m_{s_{t-1}s_t}^{(a_{t-1},a_t)}m_{r_{t-1}r_t}^{(b_{t-1},b_t)}=0$. Hence, we obtain 
$\p(a_{0..t}|r_{0..t})=\p(a_{0..t}|s_{0..t})$, which implies (\ref{200123:3}).

\medskip
``$\Leftarrow$'': Let $a,a',b,b'\in \mathcal A$ and $r,r',s,s'\in
\widehat{\mathcal R}$ such that $m_{rs}^{(a,b)}m_{r's'}^{(a',b')}\not=
m_{r's'}^{(a,b)}m_{rs}^{(a',b')}$. If $b=b'$ we choose $c\in \mathcal
A$ with $c\not=b$. Obviously, at least one of the two inequalities 
\[ \frac{m_{rs}^{(a,b)}}{m_{r's'}^{(a,b)}}\not= \frac{m_{rs}^{(a,c)}}{m_{r's'}^{(a,c)}}, \ \ \ \  \ \ 
\frac{m_{rs}^{(a',b')}}{m_{r's'}^{(a',b')}}\not= \frac{m_{rs}^{(a,c)}}{m_{r's'}^{(a,c)}}
\]
holds. Hence, we can assume that $b\not= b'$.
We now define $A_0$ and $R_0$ as independent variables specified by 
\[ \p(A_{0}=a)=\p(A_{0}=a')=\frac{1}{|\{a,a'\}|},\ \ \
  \p(R_{0}=r)=\p(R_{0}=r')=\frac{1}{|\{r,r'\}|}. \]
Furthermore, it is assumed that the transition matrix $M$ satisfies
$m_{ab}=m_{a'b'}=1/2$ if $a= a'$ and $m_{ab}=m_{a'b'}=1$ if $a\not= a'$.

In both cases, $\p(A_0=a,A_1=b)=\p(A_0=a',A_1=b')=1/2$ and Bayes
formula (\ref{231117:2}) specializes to  
\[
\p(A_0=a,A_1=b|R_0=r,R_1=s)=\frac{m_{rs}^{(a,b)}}{m_{rs}^{(a,b)}+m_{rs}^{(a',b')}}.
\]
Since $m_{rs}^{(a,b)}(m_{r's'}^{(a,b)}+m_{r's'}^{(a',b')})\not=
m_{r's'}^{(a,b)}(m_{rs}^{(a,b)}+m_{rs}^{(a',b')})$ we  obtain
\begin{eqnarray*}
\p(A_0=a,A_1=b|R_0=r,R_1=s) & = & \frac{m_{rs}^{(a,b)}}{m_{rs}^{(a,b)}+m_{rs}^{(a',b')}} \\
& \not= & \frac{m_{r's'}^{(a,b)}}{m_{r's'}^{(a,b)}+m_{r's'}^{(a',b')}} \\
& = & \p(A_0=a,A_1=b|R_0=r',R_1=s').
\end{eqnarray*}
It follows from $b\not= b'$ that $\p(A_0=a,A_1=b)=\p(A_1=b)$
and therefore 
\[ \p(A_1=b|R_0=r,R_1=s) 
\not=  \p(A_1=b|R_0=r',R_1=s'). \]
Hence, (\ref{200123:3}) does not hold.
\end{proof}

\begin{corollary} ~\label{200322:c1}
The rating process $R_t$ has the Markov property if
 \begin{enumerate}[(i)]
  \item the variables $A_0$ and $R_0$, which are specified by their
    joint distribution $\lambda\in [0,1]^{KJ}$, are independent  and 
  \item the conditional transition matrices $M^{(a,b)}$ have identical
    ratios of non-default migration probabilities. 
  \end{enumerate}
\end{corollary}
\begin{proof} By Theorem \ref{200126:t1}, Equation (\ref{200123:3}) holds for $t\in \Nzero$,
which implies
\begin{equation*}
\label{eq:9}
\p(A_t=a|R_0=r_0, \ldots, R_t=r_t) = \p(A_t=a|R_t=r_t)
\end{equation*}
for $a\in \mathcal A$ and
$r_0,\ldots,r_t\in \widehat{\mathcal R}$ with
$\p(R_0=r_0,\ldots,R_t=r_t)>0$. Hence, in case of no default at time
$t$, i.e., $R_t\not= \overline{r}$, we obtain 
     \begin{multline*}
    \p(R_{t+1}=r_{t+1}|R_0=r_0,\ldots,R_t=r_t)\\
    \begin{aligned}
      & = \sum_{a,b\in \mathcal A} \p(R_{t+1}=r_{t+1},
        A_t=a,A_{t+1}=b|
        R_0=r_0, \ldots, R_t=r_t)\\
      &=  \sum_{a,b\in \mathcal A} \p(R_{t+1}=r_{t+1},
        A_{t+1}=b| A_t=a,
        R_0=r_0, \ldots, R_t=r_t)\, \p(A_t=a|R_0=r_0,\ldots, R_t=r_t)\\
      & =
      \sum_{a,b\in \mathcal A} \p(R_{t+1}=r_{t+1},A_{t+1}=b|R_t=r_t,A_t=a)  \p(A_{t}=a|R_t=r_t) \\
      & = \p(R_{t+1}=r_{t+1}|R_t=r_t).
    \end{aligned}
   \end{multline*}
If $R_t=\overline{r}$ then the Markov property holds because $\overline{r}$ is absorbing.
\end{proof}

\subsection{Time-inhomogeneous Markov chains on ratings}
\label{sec:time-inhom-mark}
As shown above, rating processes typically do not have the Markov
property. However, it is easy to prove that for each rating process
$(R_t)_{t\in \Nzero}$, a time-inhomogeneous Markov process can be
constructed that replicates the unconditional distribution of
$(R_t)_{t\in \Nzero}$ for each $t\in \Nzero$.\footnote{A Markov chain
  is time-inhomogeneous if the transition matrices are time-dependent,
  see e.g.\ Chapter 1 of \cite{Bremaud1999}.}
Using mathematical results from Markov chain theory, this
process will then be used to determine the long-term and asymptotic
rating transition behaviour of $(R_t)_{t\in \Nzero}$, see Section
\ref{sec:asympt-prop}. 

More precisely, we define a time-inhomogeneous Markov
chain on the state space $\mathcal R$ that
replicates the rating distribution
\begin{equation*}
  (\p(R_t= r))_{r\in{\mathcal R}}, \quad t\in \Nzero.
\end{equation*}

\begin{proposition} \label{200503:p1}
  There exists a time-inhomogeneous Markov process
  $({\widetilde R}_t)_{t\in\Nzero}$ with
  \begin{enumerate}[(i)]
  \item state space $\mathcal R$,
  \item initial distribution ${\tilde\lambda}$ on $[0,1]^J$ defined by
    ${\tilde\lambda}=(\p(R_0= r))_{r\in{\mathcal R}}$,
  \item transition matrices $({\widetilde P}_t)_{t\in \Nzero}$ that can be
    determined from $(\lambda,P)$, such that
    \begin{equation}
      \p({\widetilde R}_t=r) = \p(R_t=r), \quad t\in \Nzero, \quad r\in{\mathcal R}.
      \label{180207:1}
    \end{equation}
  \end{enumerate}
\end{proposition}
In particular, long-term default probabilities and distributions of
non-default ratings can be immediately derived via 
\begin{equation}
  (\p(R_T= r))_{r\in{\mathcal R}} = {\tilde\lambda} \cdot
  \prod_{t=1}^T {\widetilde P}_t, \quad T\in\Nzero. \label{240210:1}
\end{equation}

Essentially, the distributions of the process ${\widetilde R}$ are
averaged with respect to the economic state transitions. Depending on 
the initial distribution $\lambda$ of $(X_t)_{t\in
  \Nzero}=(A_t,R_t)_{t\in \Nzero}$, information about the initial
economic state may enter the rating distributions, which creates the
time-inhomogeneity in the transitions.   
\begin{proof}
  First, we construct transition matrices $({\widetilde P}_t)_{t\in \Nzero}$. For
  $t\in \Nzero$, let $\lambda_t\in [0,1]^{KJ}$ be the distribution of
  $X_t$, i.e., $ \lambda_{t,(a,r)}= \p(X_t=(a,r))$ for $(a,r)\in
  \mathcal X$, 
and define the matrix
  ${\widetilde P}_t=({\tilde p}_{t,rs})\in [0,1]^{J\times J}$ for
  $t\geq 1$: if 
  $\p(R_{t-1}=r)>0$, then
  \begin{align*}
    {\tilde p}_{t,rs}
    &= \p(R_t=s|R_{t-1}=r) %
      = \sum_{a\in \mathcal A} \p(R_t=s|R_{t-1}=r,A_{t-1}=a) \cdot
      \p(A_{t-1}=a|R_{t-1}=r)\\
    &= \sum_{a\in \mathcal A} \p(R_t=s|R_{t-1}=r,A_{t-1}=a) \cdot
      \frac{\p(A_{t-1}=a,R_{t-1}=r)} {\p(R_{t-1}=r)}\\
    &= \frac{\sum_{a,b\in \mathcal A} p_{(a,r)(b,s)} \cdot
      \lambda_{t-1,(a,r)}} {\sum_{a'\in \mathcal A} \lambda_{t-1,(a',r)}}. %
  \end{align*}

  Otherwise, ${\tilde p}_{t,rs}=0$ for $s\in\widehat{\mathcal R}$ and
  ${\tilde p}_{t,r \overline{r}}=1$.  By definition, the matrices
  ${\widetilde P}_t$, $t\geq 1$, 
  are stochastic and can be calculated iteratively from
  $(\lambda, P)$.  Hence, the process $({\widetilde R}_t)_{t\geq\Nzero}$
  can be defined as a time-inhomogeneous Markov chain. It is easily
  verified that $\p({\widetilde R}_t=r)=\p(R_t=r)$, $t> 0$: From
  \begin{equation*}
    \p(R_t=s) = \sum_{r\in {\mathcal R}} \p(R_t=s|R_{t-1}=r)\cdot
    \p(R_{t-1}=r), 
  \end{equation*}
  the claim follows by definition of ${\tilde\lambda}$ and by
  induction.
\end{proof}

\citet{Bluhm2007} propose a parametric class of time-inhomogeneous
Markov chains to calibrate observed multi-year default
rates. Proposition \ref{200503:p1} provides theoretical support to the
claim that a convincing fit can be achieved even if the underlying
process is not Markovian.

\section{Asymptotic properties}
\label{sec:asympt-prop}

In this section, we use standard techniques from the theory of Markov
chains to prove that the sequence of transition matrices 
$({\widetilde P}_t)_{t\in \Nzero}$ converges to a limit, denoted by
${\widetilde P}_{\infty}$. The proof is based on the Perron-Frobenius
Theorem (\citet{Perron1907}; \citet{Frobenius1912}).  
In order to apply the theorem we assume in this section that the
economic state process $(A_t)_{t\in \Nzero}$ is irreducible and
aperiodic. Hence, the transition matrix $M$ is primitive, i.e., there
exists an integer $k$ such that $M^k>0$. Furthermore, recall that
$m^{(a,b)}_{rs}>0$ for all $r\in\widehat{\mathcal R}$, $s\in{\mathcal
  R}$ and $a,b\in\mathcal A$. This 
assumption implies that $\widehat{P}$ is a primitive matrix, where
$\widehat{\mathcal X}=\mathcal A\times \widehat{\mathcal R}$ and
$\widehat{P}$ is the restriction of $P$ to $\widehat{\mathcal X}$,
i.e., 
$\widehat{P}=(\hat{p}_{(a,r),(b,s)})\in [0,1]^{K(J-1)\times K(J-1)}$ with
$\hat{p}_{(a,r),(b,s)}=p_{(a,r),(b,s)}$ for
$(a,r),(b,s)\in \widehat{\mathcal X}$. Likewise, we denote the restriction of
$\lambda$ to $\widehat{\mathcal X}$ by $\hat{\lambda}\in
[0,1]^{K(J-1)}$. Note that the assumption $\p(X_0\in \widehat{\mathcal X})=1$ implies that
$\hat{\lambda}$ is a distribution on $\widehat{\mathcal X}$.

The matrix $\widehat{P}$ is not stochastic but it is primitive. Hence,
the following properties are a consequence of the Perron-Frobenius 
Theorem (see e.g.\ Theorems 1.1 and 1.2 of \cite{Seneta1973}, Chapter 6
of \cite{Bremaud1999}):
\begin{enumerate}
\item [(P1)] The primitive matrix $\widehat{P}$ has a real eigenvalue
  $\rho$ with algebraic as well as geometric multiplicity one
  such that $\rho>0$ and $\rho>|\rho'|$ for any other
  eigenvalue $\rho'$. Moreover, the left eigenvector $\mu$ and
  the right eigenvector $\nu$ associated with $\rho$ can be
  chosen positive and such that $\mu \nu=1$.\footnote{Note that the
    left eigenvector $\mu$ is a row-vector and the right
    eigenvector $\nu$ is a column-vector, both of dimension 
    $K(J-1)$. Hence, $\mu \nu$ is a real number whereas $\nu \mu$ is a
    $K(J-1)\times K(J-1)$ matrix.}
\item [(P2)] The matrix $((1/\rho)\widehat{P})^n$ converges to the
  matrix $\nu \mu$. 
\end{enumerate}
We assume that the components of $\mu$ add up to 1, i.e., $\mu$ is a
distribution on $\widehat{\mathcal X}$. Let ${\tilde\mu}$ be the
corresponding distribution on $\mathcal R$ specified by
\begin{equation} \label{200327:1}
  {\tilde\mu}_r=\sum_{a\in \mathcal A}\mu_{(a,r)} \text{ for }
 r\in\widehat{\mathcal R} \quad\text{ and } \quad {\tilde\mu}_{\overline{r}}=0, 
\end{equation}
and define the stochastic matrix
${\widetilde P}_{\infty}=({\tilde p}_{\infty,rs})\in [0,1]^{J\times J}$ by
\begin{equation}
  \label{eq:2}
  {\tilde p}_{\infty,rs}=\sum_{a\in \mathcal A}\mu_{(a,r)}\sum_{b\in
    \mathcal A}
  p_{(a,r),(b,s)}/{\tilde\mu}_{r},\quad r\in\widehat{\mathcal R},
   s\in{\mathcal R}, 
\end{equation}
${\tilde p}_{\infty,\overline{r} r}=0$ for $r\in\widehat{\mathcal R}$ and
${\tilde p}_{\infty,\overline{r}\, \overline{r}}=1$.

We will now study the asymptotic properties of the
Markov chain $(X_t)_{t\in\Nzero}$. Since ${\mathcal
  A}\times\{\overline{r}\}$ is a closed subset of the state space of
$(X_t)_{t\in\Nzero}$ it immediately follows from the assumptions on
the transition matrices $M$ and $M^{(a,b)}$ that 
\begin{equation*}
    \lim_{t\rightarrow\infty}\p(X_t \in {\mathcal A}\times\{\overline{r}\})=1.
  \end{equation*}
The more relevant question is the asymptotic behaviour of
$(X_t)_{t\in\Nzero}$ conditionally on no default, which is summarized
in the following theorem. In particular, it is shown that 
the conditional distribution of non-default rating states of $X$
converges to $\mu$, which implies convergence of $({\widetilde P}_t)_{t\in \N}$ to
${\widetilde P}_{\infty}$.  Furthermore, $1-\rho$ is the limit of the
marginal default rates.\footnote{Marginal default rates refer to
  probabilities of default conditional on no prior default, as opposed
  to cumulative default rates.} 

In the following, note that $\{R_t=\overline{r}\}=\{X_t\in
(\cdot,\overline{r})\}$ and $\{R_t\not=\overline{r}\} =
\{X_t\not\in(\cdot,\overline{r})\}$. For ease of exposition, we use
both notations. 

\begin{theorem}
  \label{170621:t2}
\begin{enumerate}[(i)]
\item  For all $(a,r)\in \widehat{\mathcal X}$,
  \begin{equation*}
    \lim_{t\rightarrow\infty}\p(X_t =(a,r)| X_t\not\in(\cdot,\overline{r}))=\mu_{(a,r)}.
  \end{equation*}
  Furthermore,
  \begin{equation} \label{200521:1}
    \lim_{t\rightarrow\infty} {\widetilde P}_t={\widetilde P}_{\infty}\quad\text{ and } \quad
    \lim_{t\rightarrow\infty} \p(R_t=\overline{r}| R_{t-1}\not=
    \overline{r})=1-\rho.
  \end{equation}
\item  If the limit distribution $\mu=(\mu_{(a,r)})\in [0,1]^{K(J-1)}$ is
  the distribution of $X_0$ on the set of non-default states
  $\widehat{\mathcal X}$, i.e., $\hat{\lambda}=\mu$, then for all
  $t\in \N$
  \begin{equation*}
    \p(X_t =(a,r)|X_t\notin (\cdot,\overline{r}))=\mu_{(a,r)}\quad\text{
      and } \quad {\widetilde P}_t={\widetilde P}_{\infty}. 
  \end{equation*}
\end{enumerate}
\end{theorem}

\begin{proof}
  (i) Since $\nu$ is a column-vector with positive entries and
  $\hat{\lambda}\not= 0$ has only non-negative entries, the vector
  product $c=\hat{\lambda} \cdot \nu$ is a positive number.  By (P2),
  \begin{equation}
    \lim_{t\rightarrow\infty} \hat{\lambda}\cdot
    ((1/\rho)\widehat{P})^t=\hat{\lambda} \cdot \nu \cdot \mu=c\cdot
    \mu. \label{170307:1} 
  \end{equation}
  Let $t\in\N$ and define $\alpha_t\in [0,1]^{K(J-1)}$ as the
  distribution of $X_t$ conditional on $R_t\not=\overline{r}$, i.e.,
  \begin{equation*}
    \alpha_t=\p(X_t=(a,r)| X_t\not\in(\cdot,\overline{r}))_{(a,r)\in\widehat{\mathcal X}}. 
  \end{equation*}
  Since the set of default states ${\mathcal A}\times\{\overline{r}\}$ is closed,
  \begin{equation}
    \alpha_t=\frac{\p(X_{t-1}\notin (\cdot,\overline{r}))}{\p(X_t\notin
      (\cdot,\overline{r}))}\cdot 
    \alpha_{t-1}\cdot \widehat{P}=\frac{1}{\p(X_t\notin (\cdot,\overline{r}))}\cdot 
    \hat{\lambda}\cdot\widehat{P}^t  \label{170309:1} 
  \end{equation}
  and therefore
  \begin{equation*}
    \alpha_t=\frac{\rho^t}{\p(X_t\notin (\cdot,\overline{r}))}\cdot
    \hat{\lambda}\cdot ((1/\rho)\widehat{P})^t.
  \end{equation*}
  Using (\ref{170307:1}) and the fact that the $L_1$-norm of $\mu$ and
  $\alpha_t$ is 1 for each $t\in\N$ we  obtain
  \begin{equation}
    \lim_{t\rightarrow\infty} \alpha_t=\mu. \label{170307:2}
  \end{equation}
Since $\widehat{P}$ is primitive, $\p(X_t=(a,r))>0$ for all
  $(a,r)\in\widehat{\mathcal X}$ and sufficiently large $t\in\N$. By the definition of the matrix
  ${\widetilde P}_t=({\tilde p}_{t,rs})\in [0,1]^{J\times J}$, for all $r\in\widehat{\mathcal R}$ and $s\in{\mathcal R}$
  \begin{align}
    {\tilde p}_{t+1,r s} & =  \left( \sum_{a\in A}\p(X_t=(a,r)) \sum_{b\in A} p_{(a,r),(b,s)}\right) /\p(X_t\in (\cdot,r))\nonumber \\
                     & =  \left( \sum_{a\in A}(\p(X_t=(a,r))/\p(X_t\notin (\cdot,\overline{r}))) \sum_{b\in A} p_{(a,r),(b,s)}\right) /(\p(X_t\in (\cdot,r))/\p(X_t\notin (\cdot,\overline{r}))) \nonumber\\
                     & = \left( \sum_{a\in A} \alpha_{t,(a,r)}
                       \sum_{b\in A} p_{(a,r),(b,s)}\right)
                       /\sum_{a\in A} \alpha_{t,(a,r)} 
                       \label{170619:3}
  \end{align}
  and ${\tilde p}_{t+1,\overline{r}r}=0$ for $r\in\widehat{\mathcal R}$  and
  ${\tilde p}_{t+1,\overline{r}\, \overline{r}}=1$. Hence, by (\ref{170307:2}),
  \begin{equation*}
    \lim_{t\rightarrow\infty} {\widetilde P}_t={\widetilde P}_{\infty}.
  \end{equation*}
  From (\ref{170307:2}) and the first equation in (\ref{170309:1}) it
  follows that
  \begin{equation*}
    \mu=\left(\lim_{t\rightarrow\infty}\frac{\p(X_{t-1}\notin
        (\cdot,\overline{r}))}{\p(X_t\notin (\cdot,\overline{r}))}\right) \cdot \mu\cdot 
    \widehat{P}.
  \end{equation*}
  Since $\mu$ is an eigenvector with eigenvalue $\rho$,
  \begin{equation*}
    \lim_{t\rightarrow\infty}\frac{\p(X_{t-1}\notin (\cdot,\overline{r}))}{
      \p(X_t\notin (\cdot,\overline{r}))}=\frac{1}{\rho}\quad\text{ and }\quad
    \lim_{t\rightarrow\infty} \p(X_t\in (\cdot,\overline{r}) |X_{t-1}\notin
    (\cdot,\overline{r}))=1-\rho.
  \end{equation*}
  (ii) We now assume that $\hat{\lambda}=\mu$. It follows from
  (\ref{170309:1}) and $\mu\cdot\widehat{P}=\rho\cdot \mu$ that
  \begin{equation*}
    \alpha_t=\frac{1}{\p(X_t\notin (\cdot,\overline{r}))}\cdot
    \mu\cdot\widehat{P}^t=\frac{1}{\p(X_t\notin
      (\cdot,\overline{r}))}\cdot\rho^t\cdot \mu=\mu 
  \end{equation*}
  for every $t\in\N$ and therefore, by (\ref{170619:3}),
  ${\widetilde P}_t={\widetilde P}_{\infty}$.
\end{proof}

\begin{definition} \label{230102:d1}
Let $({S}_t)_{t\in \Nzero}$ be the time-homogeneous
Markov chain with state space $\mathcal R$, transition matrix
${\widetilde P}_{\infty}=({\tilde p}_{rs})\in [0,1]^{J\times J}$ and
initial distribution ${\tilde\lambda}$. Then $({S}_t)_{t\in \Nzero}$
is called the {\em asymptotic approximation\/} of $(R_t)_{t\in \Nzero}$. 
\end{definition}

The following corollary summarises the properties of the asymptotic approximation.

\begin{corollary} \label{200503:c2}
\begin{enumerate}[(i)]
\item For any $r\in\widehat{\mathcal R}$,
  \begin{equation} \label{231229:1}
    \lim_{t\rightarrow\infty}\p({S}_t =r| {S}_t\not= \overline{r})=
    \lim_{t\rightarrow\infty} \p(R_t = r| R_t\not = \overline{r}).
  \end{equation}
\item If the limit distribution
$\mu=(\mu_{(a,r)})\in [0,1]^{K(J-1)}$ is the distribution of $X_0$ on
$\widehat{\mathcal X}$  then 
\begin{equation} \label{200503:1}
  \p({S}_t =r) = \p(R_t=r) =\p(X_t\in (\cdot,r))
\end{equation}
for $t\in \Nzero$ and $r\in{\mathcal R}$. Furthermore, for $t\in{\mathbb N}$,
\begin{equation} \label{200503:2}
  \p({S}_t =\overline{r}|\p({S}_{t-1} \not=\overline{r})
= 
\p(R_t =\overline{r}|\p({R}_{t-1} \not=\overline{r})
= (1-\rho)\rho^{t-1}.
\end{equation}
\end{enumerate}
\end{corollary}
\begin{proof} 
We first assume that the limit distribution
$\mu=(\mu_{(a,r)})\in [0,1]^{K(J-1)}$ is the distribution of $X_0$ on
$\widehat{\mathcal X}$. It follows from Theorem \ref{170621:t2} that 
${\widetilde R}_{t\in \Nzero}$ is a time-homogeneous Markov process,
which equals the asymptotic approximation $({S}_t)_{t\in \Nzero}$, and 
\[ \p(R_t\not=\overline{r}| R_{t-1}\not= \overline{r})=\rho\ \ \ \
  \hbox{ and }\ \ \ \  \p(R_t=\overline{r}| R_{t-1}\not=
  \overline{r})=1-\rho \] 
for all $t\in {\mathbb N}$. Hence, Equations (\ref{200503:1}) and
(\ref{200503:2}) follow from Proposition \ref{200503:p1}. Furthermore,
we obtain from Theorem \ref{170621:t2} that $({\tilde\mu}_r)_{r\in
  \widehat{\mathcal R}}$ is a positive left eigenvector of the
restriction of ${\widetilde P}_{\infty}$ to 
  $\widehat{\mathcal R}$. Analogously to the proof of Theorem
  \ref{170621:t2}, it follows from the Perron-Frobenius 
Theorem that
\begin{equation}
    \lim_{t\rightarrow\infty}\p({S}_t =r| {S}_t\not= \overline{r})={\tilde\mu}_{r}, \quad r\in \widehat{\mathcal R},
    \label{170310:2}
  \end{equation}
holds for any initial distribution  of ${S}_0$. Hence, (\ref{231229:1}) follows from Theorem \ref{170621:t2}.
\end{proof}

\section{Rating philosophies: Point-in-time and through-the-cycle}
\label{sec:rating-philosophies}

\subsection{Definition of TTC and PIT ratings}
\label{sec:definition-ttc-pit}

This section deals with the characterisation of rating
systems. Ratings are designed to provide information on the credit
quality of an obligor. However, rating systems may differ with respect
to the type of information reflected in the rating. There exist two
main rating philosophies:
\begin{enumerate}[(i)]
\item Point-in-time (PIT) ratings reflect the likelihood of default
  over a future period, e.g. one year, and are therefore dependent on 
  firm-specific as well as systematic factors. The PIT rating of an
  obligor reflects the current state of the economy as well as all
  available information about the current and future economic
  development. As such, PIT ratings fluctuate over economic cycles
  (e.g.\ economic downturns lead to an increased number of
  downgrades), whereas expected one-period default probabilities of
  PIT rating classes are stable.  
\item In contrast, through-the-cycle (TTC) ratings have a longer
  rating horizon covering multiple periods, i.e., the rating is
  performed through 
the cycle. As a consequence, economic effects cancel out and TTC ratings are therefore
independent of the state of
  the economy and, as such, are insensitive to economic or credit
  cycles. However, since actual corporate default rates depend on the
  economic environment the expected one-period default probabilities 
  of TTC rating classes inevitably fluctuate with the economic
  cycle.\footnote{The characteristics of expected one-period default probabilities
  of TTC rating classes is different to TTC PDs defined in the literature as long-run 
  averages of one-year default probabilities
  (e.g.\ \cite{Valles2006,Gordy2004,Altman2004}) or 
  stressed default probability (e.g.\
  \cite{Heitfield2005,Crouhy2001}). A comparison of these concepts in
  our mathematical framework is provided in Sections \ref{sec:example}
  and \ref{sec:conclusion} of this paper.} 
\end{enumerate}

\begin{figure}[t]
  \centering
  \includegraphics[width=.5\textwidth]{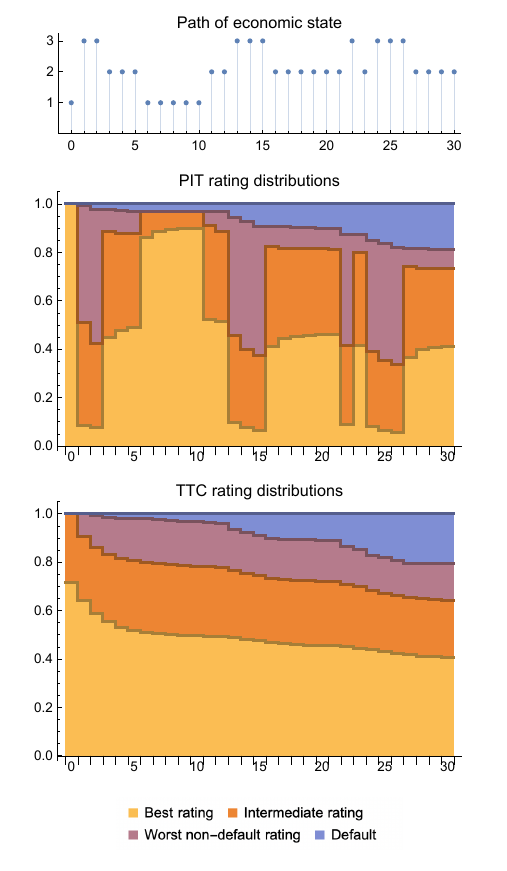}
  \caption{Example paths of economic states (top), PIT (middle) and
    TTC (bottom) rating distributions over 30 time periods for a firm
    starting in the best PIT rating. Economic states are ordered from
    best (1) to worst (3). }
  \label{fig:PITvTTC}
\end{figure}
Figure \ref{fig:PITvTTC} shows an example of an economic path (top)
together with rating distributions of PIT and TTC ratings for a firm
starting in the best PIT rating. This shows how PIT ratings fluctuate
with the economy while TTC non-default ratings are insensitive to the
economic environment. The rating models used in this
illustration will be discussed in Section \ref{sec:stylised-example}. 

Both rating philosophies are interesting theoretical concepts. 
For practical applications, however, it is almost impossible to
implement TTC or PIT ratings in a strict sense. First of all, it is
very difficult to completely separate firm-specific components from
macroeconomic developments. It is equally challenging to implement a
rating system that comprehensively reflects the current state of the
economy as well as all available information about the future economic
development in a timely manner. As a consequence, real-world rating
systems are typically hybrids in the following sense: Individual
ratings are impacted by macroeconomic developments, e.g.\ upgrades are
more frequent in economic upturns, but not all the available economic
information is reflected leading to fluctuations of default rates of
rating classes across credit cycles. Rating systems mainly differ
with respect to the weights of TTC and PIT components, i.e., whether
they are predominantly through-the-cycle or point-in-time, 
see \citet{Morone2009,Rubtsov2016} for estimating the degree of 
PIT-ness of a rating system.
The consistent calculation of TTC and PIT default probabilities is
intensively discussed in the literature, e.g.\ in 
\citet{Aguais2008,Carlehed2012,Rubtsov2016,Miu2017,Rubtsov2021}.

In order to formally define point-in-time rating migration processes
we consider expected (one-period) default probabilities for each
combination of economic state $a$ and rating $r$: 
\[ {\rm PD}(a,r)=\sum_{b\in \mathcal A} p_{(a,r)(b,\overline{r})}=
\p(R_{t+1}= \overline{r}\mid A_t=a,R_t=r) \]
for any $t\in {\mathbb N}_0$. Since the PIT rating $r$ of a company 
already reflects the impact of the state of the economy on its
creditworthiness the corresponding default probability is completely
determined by $r$. 

\begin{definition}
  \label{def:PIT}
  A rating migration process is {\em point-in-time (PIT)} if ${\rm
    PD}(a,r)={\rm PD}(b,r)$ for all ratings $r\in{\mathcal R}$ and  
economic states $a,b\in {\mathcal A}$.
\end{definition}

Essentially, this captures that ratings
  are assigned such that expected default and survival probabilities
  of a rating class do not depend on the economic state. However,
  ex-post default frequencies  
$\p(R_{t+1}=\overline{r}|R_t=r,A_t=a,A_{t+1}=b)$ still depend on the economic
state transition, e.g.\ if the economy transitions into a worse state,
more defaults will be observed compared to a neutral or positive
transition.

In contrast, through-the-cycle ratings are insensitive to economic or credit
cycles. More formally, let $Q^{(a,b)}\in [0,1]^{J\times J}$ be the
matrix of transition probabilities from rating $r$ to rating $s$
conditionally on economic states $a,b\in {\mathcal A}$ and no default: 
\begin{equation}
  q_{rs}^{(a,b)} =
  \begin{cases}
    \p(R_{t+1}=s|R_t=r,A_t=a,A_{t+1}=b,R_{t+1}\not=\overline{r}) =
    \displaystyle\frac{m_{rs}^{(a,b)}} {1-m_{r\overline{r}}^{(a,b)}}, &\text{
      if }
    r,s\not=\overline{r},\\
    1,&\text{ if } r=s=\overline{r},\\
    0,&\text{ else.}
  \end{cases}\label{220827:1}
\end{equation}

\begin{definition}
  \label{def:QTTC}
  A rating migration process is {\em through-the-cycle (TTC)} if  there
  exists a stochastic matrix $Q\in [0,1]^{J\times J}$ such that
  $Q^{(a,b)}=Q$ for all $a, b\in \mathcal A$. 
\end{definition}

The definition captures that, conditional on no-default, rating
transitions are independent of the economic state.

\subsection{Decomposition into rating and default components}
\label{sec:decomp-into-rating}
For all $a, b\in \mathcal A$ the rating migration matrix $M^{(a,b)}$
can be decomposed into a matrix product  
\begin{equation} \label{240227:1}
  M^{(a,b)} = D^{(a,b)}\cdot Q^{(a,b)},
\end{equation}
where the non-default component $Q^{(a,b)}$ has already been specified
in (\ref{220827:1}) and the default component $D^{(a,b)}\in
[0,1]^{J\times J}$ is defined as a stochastic matrix with entries 
 \begin{align*}
 d_{r\overline{r}}^{(a,b)}&= \p(R_{t+1}=\overline{r}|R_t=r,A_t=a,A_{t+1}=b) =
                      m_{r\overline{r}}^{(a,b)},\\[10pt]
    d_{rr}^{(a,b)} &=    \p(R_{t+1}\not=\overline{r}|R_t=r,A_t=a,A_{t+1}=b) =
                     1-m_{r\overline{r}}^{(a,b)}, \quad \text{ if } r\not=\overline{r},
  \end{align*}
  and all other entries taking value zero. The default component
  $D^{(a,b)}$ can be further decomposed as 
\begin{equation*}
D^{(a,b)} = D^a\cdot C^{(a,b)},
\end{equation*}
where $D^a\in
  [0,1]^{J\times J}$ is fully specified by the default probabilities
  ${\rm PD}(a,r)$, i.e., $D_a$ depends on the current economic state
  $a$ of the economy but not on futures states, whereas $C^{(a,b)}\in
  [0,1]^{J\times J}$ specifies the ex-post deviation from the expected
  default probability ${\rm PD}(a,r)$ if the economy transitions from
  $a$ to $b$. 
More precisely,
\begin{itemize}
\item for each state $a\in \mathcal A$ of the economy, $D^a=(d_{rs}^a)\in
  [0,1]^{J\times J}$ is a stochastic matrix containing default and survival 
  probabilities unconditional on the future state of the economy; its
  entries are, for  $r\not=\overline{r}$, 
\[ d_{r\overline{r}}^a ={\rm PD}(a,r) = \p(R_{t+1}=\overline{r}|A_t=a,R_t=r),\ \ \ \ \ \ 
   d_{rr}^a = 1-d_{r\overline{r}}^a,\ \ \ \ \ \ 
  d_{\overline{r}\,\overline{r}}^a=1; \]
all other entries take value zero;
\item $C^{(a,b)}\in
  [0,1]^{J\times J}$ has entries, for $r\not=\overline{r}$,
\[ c_{rr}^{(a,b)} = \frac{d_{rr}^{(a,b)}} {d_{rr}^a},\ \ \ \ \ \ 
  c_{r\overline{r}}^{(a,b)} = 1-c_{rr}^{(a,b)},\ \ \ \ \ \ 
  c_{\overline{r}\,\overline{r}}^{(a,b)}=1; \]
all other entries in $C^{(a,b)}$ are zero. 
\end{itemize}

\begin{proposition} \label{prop:matrix-decompositionD}
For $a,b\in \mathcal A$,
\begin{equation*}
M^{(a,b)} = D^{(a,b)}\cdot Q^{(a,b)}\ \ \ \hbox{and}\ \ \ D^{(a,b)} = D^a\cdot C^{(a,b)}.
\end{equation*}
The matrices $C^{(a,b)}$ satisfy 
\[ \sum_{b\in \mathcal A} m_{ab}  C^{(a,b)}=I \]
for all $a\in \mathcal A$ and identity matrix $I\in [0,1]^{J\times J}$. 
\end{proposition}

\begin{proof}
The equation $M^{(a,b)} = D^{(a,b)}\cdot Q^{(a,b)}$ immediately
follows from the definition of $D^{(a,b)}$ and $Q^{(a,b)}$. 

All matrix elements of $D^a\cdot C^{(a,b)}$ and $D^{(a,b)}$ are 0
except the diagonal elements and the last columns. By definition, 
\begin{equation}
d_{rr}^a c_{rr}^{(a,b)} = d_{rr}^{(a,b)}
\label{200307:1}
\end{equation}
and therefore the matrices $D^a\cdot C^{(a,b)}$ and $D^{(a,b)}$ have
the same diagonal elements. The matrix elements in the last column
satisfy the following equations: 
\[ d_{rr}^a c_{r\overline{r}}^{(a,b)}+d_{r\overline{r}}^a =
d_{rr}^a (1-c_{rr}^{(a,b)})+d_{r\overline{r}}^a =
d_{rr}^a -d_{rr}^{(a,b)}+d_{r\overline{r}}^a =
d_{r\overline{r}}^{(a,b)} \]
and therefore $D^{(a,b)} = D^a\cdot C^{(a,b)}$.

It follows from (\ref{200307:1}) and
\[ d_{rr}^a = \frac{\p(A_t=a,R_t=r,R_{t+1}\not=\overline{r})}{\p(A_t=a,R_t=r)},\ \ \ \ \ \ 
 d_{rr}^{(a,b)} = \frac{\p(A_t=a,A_{t+1}=b,R_t=r,R_{t+1}\not=\overline{r})}{\p(A_t=a,A_{t+1}=b,R_t=r)}  
\]
that
\[ c_{rr}^{(a,b)} =  \frac{\p(A_{t+1}=b|A_t=a,R_t=r,R_{t+1}\not=\overline{r})}{\p(A_{t+1}=b|A_t=a,R_t=r)}. \]
By Lemma \ref{200125:l1},
\[ \sum_{b\in \mathcal A} m_{ab}  c_{rr}^{(a,b)} =
\sum_{b\in \mathcal A} \p(A_{t+1}=b|A_t=a,R_t=r) c_{rr}^{(a,b)} =1 \]
for all $a\in \mathcal A$. Since all matrix elements of $C^{(a,b)}$
are 0 except the diagonal and the last column and each row sum equals
1 we obtain 
\[ \sum_{b\in \mathcal A} m_{ab}  C^{(a,b)}=I. \]
\end{proof}

PIT and TTC rating migration processes have a rather intuitive
interpretation in the context of the decomposition   
\begin{equation*}
    M^{(a,b)} = D^{a}\cdot C^{(a,b)}\cdot Q^{(a,b)}, \quad a,b\in \mathcal
    A, 
  \end{equation*}
given in the proposition above:
\begin{enumerate}
\item  A rating migration process is {\em point-in-time (PIT)} iff
  there exists a stochastic matrix $D\in [0,1]^{J\times J}$ such that
  $D^{a}=D$ for all $a \in\mathcal A$, i.e.,
  \begin{equation*}
    M^{(a,b)} = D\cdot C^{(a,b)}\cdot Q^{(a,b)}, \quad a,b\in \mathcal
    A. 
  \end{equation*}
Note that the matrix $D$, which does not depend on economic states,
specifies the expected default and survival probabilities for each
rating class whereas the economic-state dependent matrix $C^{(a,b)}$
reflects that ex-post default frequencies
$\p(R_{t+1}=\overline{r}|R_t=r,A_t=a,A_{t+1}=b)$ still depend on the
economic state transition.
\item
  A rating migration process is {\em through-the-cycle (TTC)} if there 
  exists a stochastic matrix $Q\in [0,1]^{J\times J}$ such that
  $Q^{(a,b)}=Q$ for all $a, b\in \mathcal A$, i.e., 
  \begin{equation}
    \label{eq:10}
    M^{(a,b)} = D^{(a,b)}\cdot Q=  D^{a} \cdot C^{(a,b)}\cdot Q, \quad a, b\in \mathcal A. 
  \end{equation}
In contrast to rating migrations specified by $Q$ for all economic
states, expected default probabilities are  economic-state dependent
for TTC processes, which is reflected by the specification of the
default component $D^{a}$. 
\end{enumerate}

We will now provide an alternative definition of TTC and PIT rating
migration processes in terms of the asymptotic rating migration matrix
${\widetilde P}_{\infty}$. Analogously to the decomposition of the
conditional rating migration matrices $M^{(a,b)}$ (see Equation
(\ref{240227:1})), ${\widetilde P}_{\infty}$ can be decomposed into a
default component ${\widetilde D}_{\infty}$ and a non-default
component ${\widetilde Q}_{\infty}$ with  
\begin{equation}
    \label{240227:2}
{\widetilde P}_{\infty} = {\widetilde D}_{\infty} \cdot {\widetilde Q}_{\infty}. 
\end{equation}

\begin{theorem} ~\label{200502:c1}
 \begin{enumerate}[(i)]
  \item $(R_t)_{t\in \Nzero}$ is a TTC rating migration process if and
    only if ${\widetilde Q}_{\infty}=Q^{(a,b)}$ for each
    $a,b\in{\mathcal A}$.  
  \item $(R_t)_{t\in \Nzero}$ is a PIT rating migration process if and
    only if ${\widetilde D}_{\infty}=D^{a}$ for each $a\in{\mathcal
      A}$.  
  \end{enumerate}
\end{theorem}
Note that the condition in Theorem \ref{200502:c1}(ii) can also be
expressed as follows. 

\begin{corollary} ~\label{221115:c1}
$(R_t)_{t\in \Nzero}$ is a PIT rating migration process if and only if
the default vector of ${\widetilde P}_{\infty}$ equals the PD vector 
$({{\rm PD}(a,r))}_{r\in{\mathcal R}}$
for all economic states $a\in \mathcal A$.
\end{corollary}

The proof of Theorem \ref{200502:c1} is an immediate consequence of
the following lemma, which provides a characterisation of the rating
migration matrices ${\widetilde P}_t$ of the time-inhomogeneous Markov
process $({\widetilde R}_t)_{t\in \Nzero}$ if $(R_t)_{t\in \Nzero}$ is
PIT or TTC. 

\begin{lemma} ~\label{200218:t1}
 \begin{enumerate}[(i)]
  \item Let $(R_t)_{t\in \Nzero}$ be a TTC rating migration process
    and let $Q$ denote the uniquely defined non-default component of
    the rating transition matrices $M^{(a,b)}$. Then $Q$ is the
    non-default component of each of the rating transition matrices
    ${\widetilde P}_t$ for $t>0$.  
  \item Let $(R_t)_{t\in \Nzero}$ be a PIT rating migration process
    and write each rating transition matrix $M^{(a,b)}$ in the form  
\[ M^{(a,b)}= D\cdot C^{(a,b)}\cdot Q^{(a,b)}. \]
Then $D$ is the default component of each of the rating transition
matrices ${\widetilde P}_t$ for $t>0$.  
  \end{enumerate}
\end{lemma}

\begin{proof}
  (i) Let $r,s\in\widehat{\mathcal R}$ and denote the distribution of
  $X_t$ by $\lambda_t\in [0,1]^{KJ}$, i.e., $ \lambda_{t,(a,r)}=
  \p(X_t=(a,r))$ for $(a,r)\in \mathcal X$ and $t\in \Nzero$. Using
  the expression of $\tilde p_{t,rs}$ from the proof of Proposition
  \ref{200503:p1}, we obtain
  \begin{align*}
    \tilde{p}_{t, r s}
    &=\frac{\sum_{a, b \in \mathcal A} p_{(a, r)(b, s)} \cdot \lambda_{t-1,(a, r)}}{\sum_{a \in \mathcal A} \lambda_{t-1,(a, r)}} %
 =\frac{\sum_{a, b \in \mathcal A} m_{a, b} m_{r s}^{(a, b)} \lambda_{t-1,(a, r)}}{\sum_{a \in \mathcal A}\lambda_{t-1,(a, r)}} \\[5pt]
& =\frac{\sum_{a, b \in \mathcal A} m_{a, b} q_{rs}(1-m_{r \overline{r}}^{(a, b)}) \lambda_{t-1,(a, r)}}{\sum_{a \in \mathcal A} \lambda_{t-1,(a, r)}} %
  =q_{rs} c_r,
\end{align*}
where $c_r$ is defined by
\[ c_r=\frac{\sum_{a, b \in \mathcal A} m_{a, b}(1-m_{r
   \overline{r}}^{(a, b)}) \lambda_{t-1,(a, r)}}{\sum_{a \in \mathcal A} \lambda_{t-1,(a, r)}}.
\]
Since $\sum_{s \in \hat{R}} q_{rs} c_r+\tilde p_{t, r\overline r}
=1$, it follows, together with $\sum_{s\in \hat R} q_{rs}=1$, that 
$c_r=1-\tilde p_{t,r\overline r}$. As a consequence, $\tilde p_{t, r
  s}=q_{r s}\left(1-\tilde p_{t, r \overline r}\right)$ and therefore
$q_{r s}=\tilde p_{t, r s}/(1-\tilde p_{t, r\overline r})$.  

\medskip
\noindent
(ii) It follows from Proposition \ref{prop:matrix-decompositionD} that
for every $a\in \mathcal A$ 
\begin{equation*} 
 \sum_{b\in \mathcal A} m_{ab} D^{(a,b)}=D\sum_{b\in \mathcal A} m_{ab} C^{(a,b)}=D. 
\end{equation*}
Hence, for $r\in {\mathcal R}$,
\begin{equation} \label{231231:1}
  {\tilde p}_{t,r\overline{r}}
    = \frac{\sum_{a,b\in \mathcal A} p_{(a,r)(b,\overline{r})} 
\lambda_{t-1}(a,r)} {\sum_{a\in \mathcal A} \lambda_{t-1}(a,r)} %
= \frac{\sum_{a\in \mathcal A}
\lambda_{t-1}(a,r) \sum_{b\in \mathcal A} m_{ab} d_{r\overline{r}}^{(a,b)} } {\sum_{a\in \mathcal A} \lambda_{t-1}(a,r)}
= d_{r\overline{r}}.
\end{equation}
The default component ${\widetilde D}_t$ of ${\widetilde P}_t$ is a
stochastic matrix and all matrix element are zero except the diagonal
elements and the last column, which equals $({\tilde
  p}_{t,r\overline{r}})_{r\in {\mathcal R}}$. Hence, by
(\ref{231231:1}), ${\widetilde D}_t=D$ for every $t>0$. 
\end{proof}

We finish this subsection by showing an interesting property of TTC
rating migration processes. In Section \ref{sec:markovian}, a class of
Markovian rating processes has been identified. These processes are
characterized by the property that their rating migration matrices
$M^{(a,b)}$ have identical ratios of non-default migration
probabilities. It is easy to prove that this class of rating processes
is a subclass of TTC rating migration processes. 

\begin{definition}
 The conditional transition matrices $M^{(a,b)}$
have {\em identical survival ratios} if for all $r,s\in
\widehat{\mathcal R}$ and $a,a',b,b'\in \mathcal A$ 
  \begin{equation}
\frac{d_{rr}^{(a,b)}}{d_{ss}^{(a,b)}}=\frac{d_{rr}^{(a',b')}}{d_{ss}^{(a',b')}},
    \label{200123:2}
  \end{equation}
where $D^{(a,b)}=(d_{rs}^{(a,b)})$ is the default component of
$M^{(a,b)}$, see Equation (\ref{240227:1}). 
\end{definition}

\begin{proposition} ~\label{200126:l1}
 The following conditions are equivalent for a rating migration
 process $(R_t)_{t\in \Nzero}$. 
\begin{enumerate}[(i)]
  \item The matrices $M^{(a,b)}$ have identical ratios of non-default
    migration probabilities (\ref{200126:1}). 
  \item $(R_t)_{t\in \Nzero}$ is TTC and the matrices $M^{(a,b)}$ have
    identical survival ratios (\ref{200123:2}). 
  \end{enumerate}
\end{proposition}

\begin{proof} 
``$\Rightarrow$'':  Let $r,r',s\in \widehat{\mathcal R}$ and
$a,a',b,b'\in \mathcal A$. It follows from (\ref{200126:1}) that 
\[ 
m_{rs}^{(a,b)} d_{r'r'}^{(a',b')} = 
\sum_{s'\in\widehat{\mathcal R}} m_{rs}^{(a,b)} m_{r's'}^{(a',b')}= 
\sum_{s'\in\widehat{\mathcal R}} m_{rs}^{(a',b')} m_{r's'}^{(a,b)}=
m_{rs}^{(a',b')} d_{r'r'}^{(a,b)}.
\]
By setting $r'$ to $r$ and to $s$ respectively we obtain 
\[ 
  q_{rs}^{(a,b)}=\frac{m_{rs}^{(a,b)}}{d_{rr}^{(a,b)}}
  =\frac{m_{rs}^{(a',b')}}{d_{rr}^{(a',b')}}=q_{rs}^{(a',b')}
~\hbox{ and }~ 
\frac{d_{rr}^{(a,b)}}{d_{rr}^{(a',b')}}=
\frac{m_{rs}^{(a,b)}}{m_{rs}^{(a',b')}}=
\frac{d_{ss}^{(a,b)}}{d_{ss}^{(a',b')}}.
\]

``$\Leftarrow$'': Let $r,r',s,s'\in \widehat{\mathcal R}$ and
$a,a',b,b'\in \mathcal A$. Since $(R_t)_{t\in \Nzero}$ is TTC there
exists a matrix $Q=(q_{rs})\in [0,1]^{J\times J}$ such that
$m_{rs}^{(a,b)}=q_{rs}d_{rr}^{(a,b)}$. Together with (\ref{200123:2}), 
\[
  \frac{m_{rs}^{(a,b)}}{m_{r's'}^{(a,b)}}=\frac{q_{rs}d_{rr}^{(a,b)}}{q_{r's'}d_{r'r'}^{(a,b)}}
  =\frac{q_{rs}d_{rr}^{(a',b')}}{q_{r's'}d_{r'r'}^{(a',b')}}=\frac{m_{rs}^{(a',b')}}{m_{r's'}^{(a',b')}}. \] 
\end{proof}

\section{Ordered rating processes}
\label{sec:ordered-ratings}
\subsection{Ordered rating distributions}
In this section, we formalize the notions of ``better'' or ``worse''
economic and rating states. A rating $r$ is typically considered
better than a rating $s$ if the multi-year PDs of $r$ are consistently
smaller than the multi-year PDs of $s$. In order to satisfy this
condition, not only the 1y PDs of $r$ and $s$ have to be ordered
accordingly but the rating migration matrices also have to meet
certain requirements. Informally, the probability that an $r$-rated
company is downgraded to a bad rating has to be lower than the
corresponding migration probability for an $s$-rated company. It is
well-know that these conditions on migration probabilities can be
formalised by monotonicity properties of rating migration matrices,
where a migration matrix is called monotone with respect to a given
(partial) order on the set of rating distributions if the order is
preserved under multiplication with the respective migration
matrix. The main objective of the section is to study conditions on
the transition matrices $M$, $M^{(a,b)}$ and $P$ to ensure
monotonicity of the matrices with respect to a given order.

Our starting point is the formal specification of a total order on the
set ${\mathcal R}$ of ratings. We assume  
that the vector of asymptotic PDs ${({\tilde
    p}_{r\overline{r}})}_{r\in{\mathcal R}}$ consists of $J$ distinct values
and introduce the total order $\leq$ on ${\mathcal R}$ by $r \leq r'$
if ${\tilde p}_{r\overline{r}}\leq {\tilde p}_{r'\overline{r}}$ for
$r,r'\in{\mathcal R}$. Note that the default rating $\overline{r}$ is
the biggest element, i.e., $r\leq\overline{r}$ for all $r\in{\mathcal
  R}$.  

Consider rating distributions $\lambda$ and $\mu$ represented by their
density functions, i.e., $\lambda={(\lambda_r)}_{r\in{\mathcal R}}$
and $\mu={(\mu_r)}_{r\in{\mathcal R}}$ are elements of $[0,1]^J$ with
$\sum_{r\in{\mathcal R}}\lambda_r=\sum_{r\in{\mathcal R}}\mu_r=1$. The
distribution $\mu$ is typically considered to be more risky than
$\lambda$ if 
for all $r\in{\mathcal R}$
\begin{equation} 
\sum_{r\leq s}\lambda_s\leq\sum_{r\leq s} \mu_s, \label{231028:1}
\end{equation}
see \citet{Jarrow1997}. The partial order on rating distributions
defined by (\ref{231028:1}) is called {\em stochastic dominance} or
{\em usual stochastic order}, written
$\lambda\leq_{st}\mu$. It is obvious that 
 (\ref{231028:1}) is equivalent to
 $\bar{F}_{\lambda}(r)\leq\bar{F}_{\mu}(r)$ for all $r\in{\mathcal
   R}$, where $\bar{F}_{\lambda}(r)=\sum_{r\leq s}\lambda_s$ denotes
 the survival function of $\lambda$.
 In the context of a rating transition matrix, where each row
 represents a rating distribution, tail probabilities thus induce the
 notion of ``better'' and ``worse'' ratings. 
Another condition equivalent to (\ref{231028:1}) is
\begin{equation} 
\lambda x^T\leq \mu x^T \hbox{ for each increasing real-valued
  sequence } x=(x_1,\ldots,x_J)\in {\mathbb R}^J,  \label{231028:2} 
\end{equation}
see \citep{Mueller2002,Foellmer2002,Shaked2007,Rolski2009}. 

A stochastic matrix $L\in [0,1]^{J\times J}$ is called {\em monotone with
respect to a partial order $\leq$\/} on the set of distributions
if $\lambda L\leq \mu L$ for all distributions with
$\lambda\leq \mu$, see e.g.\ \citep{Rolski2009}. In particular, a
stochastic matrix $L$ is called {\em stochastically monotone\/} if it
is monotone with respect to $\leq_{st}$. It is easy to show (Theorem
7.4.1 of \cite{Rolski2009}) that $L$ is stochastically monotone iff 
\begin{equation} 
r\leq s \hbox{ implies } l_r\leq_{st}l_s \hbox{ for all } r,s\in
{\mathcal R},  \label{231028:3} 
\end{equation}
where the distribution $l_r$ denotes the $r$-th row of matrix $L$.

\medskip
We now apply the concept of stochastic monotonicity to the
(conditional) rating migration matrices $M^{(a,b)}\in [0,1]^{J\times
  J}$. It follows immediately from the definition that the $M^{(a,b)}$
are stochastically monotone if and only if the order relation
$\leq_{st}$ is preserved on each economic path. 

\begin{theorem} ~\label{231103:t1}
The following two conditions are equivalent:
\begin{enumerate}[(i)]
  \item  For each $a,b\in {\mathcal A}$, the (conditional) rating
    migration matrix $M^{(a,b)}\in [0,1]^{J\times J}$ is
    stochastically monotone. 
\item For all rating distributions $\lambda,\mu\in [0,1]^J$ and each
  path  
$a_0,\ldots,a_T\in {\mathcal A}$ of economic states,
\[ \lambda\leq_{st} \mu \hbox{ implies } \lambda \prod_{t=0}^{T-1}
  M^{(a_t,a_{t+1})}\leq_{st} \mu \prod_{t=0}^{T-1}
  M^{(a_t,a_{t+1})}. \] 
  \end{enumerate}
\end{theorem}

\subsection{Ordered distributions of economic states and ratings}
We will now order the set of distributions on the product space 
$\mathcal X=\mathcal A\times \mathcal R$ of economic states and
ratings. First we fix a total order $\leq$ on the set of economic
states ${\mathcal A}$. The total orders on ${\mathcal A}$ and
${\mathcal R}$ are now extended to the (partial) product order $\leq$
on ${\mathcal X}$: $(a,r)\leq (b,s)$ in ${\mathcal X}$ if $a\leq b$ in
${\mathcal A}$ and $r\leq s$ in ${\mathcal R}$.   

Let $\lambda=(\lambda_{a,r})_{(a,r)\in{\mathcal X}}$ and
$\mu=(\mu_{a,r})_{(a,r)\in{\mathcal X}}$ be distributions on
${\mathcal X}$, i.e.,  elements of $[0,1]^{KJ}$ with
$\sum_{(a,r)\in{\mathcal X}}\lambda_{a,r}=1$. The  following
definition generalizes the concept of stochastic dominance
(or the usual stochastic order) to
distributions on ${\mathcal X}$: the distribution $\mu$ {\em
  stochastically dominates\/} $\lambda$, written
$\lambda\leq_{st}\mu$, if   
\begin{equation}
\sum_{(a,r)\in U} \lambda_{a,r}\leq  \sum_{(a,r)\in U} \mu_{a,r},\ \ \
\ \ \hbox{for all upper sets } U\subseteq {\mathcal
  X}, \label{230327:2} 
\end{equation}
where $U\subseteq {\mathcal X}$ is called an upper set if $(a,r)\in U$ 
and $(a,r)\leq (b,s)$ implies $(b,s)\in U$.\footnote{
  Upper sets of the ordered set $\{1,2\}\times \{1,2,3\}$
\footnotesize{with Hasse diagram}
\tiny{\begin{tikzpicture}[baseline=10mm,scale=0.75, nodes={}]
  \node (a1b1) at (1,0) {$(1,1)$};
  \node (a1b2) at (0,1) {$(1,2)$};
  \node (a1b3) at (-1,2) {$(1,3)$};
  \node (a2b1) at (2,1) {$(2,1)$};
  \node (a2b2) at (1,2) {$(2,2)$};
  \node (a2b3) at (0,3) {$(2,3)$};

  \draw (a1b1) -- (a1b2);
  \draw (a1b2) -- (a1b3);
  \draw (a1b1) -- (a2b1);
  \draw (a1b2) -- (a2b2);
  \draw (a1b3) -- (a2b3);
  \draw (a2b1) -- (a2b2);
  \draw (a2b2) -- (a2b3);
  
\end{tikzpicture}}:
\\
  \tiny{$\{\},\left\{
\begin{array}{c}
 (2,3) \\
\end{array}
\right\},\left\{
\begin{array}{c}
 (2,3) \\
 (1,3) \\
\end{array}
\right\},\left\{
\begin{array}{c}
 (2,3) \\
 (2,2) \\
\end{array}
\right\},\left\{
\begin{array}{c}
 (2,3) \\
 (2,2) \\
(1,3) \\
\end{array}
\right\},\left\{
\begin{array}{c}
 (2,3) \\
 (2,2) \\
(2,1) \\
\end{array}
\right\},\left\{
\begin{array}{c}
(2,3) \\
 (2,2) \\
(1,3) \\
(1,2) \\
\end{array}
\right\},\left\{
\begin{array}{c}
 (2,3) \\
 (2,2) \\
(2,1) \\
(1,3) \\
\end{array}
\right\},\left\{
\begin{array}{c}
 (2,3) \\
 (2,2) \\
(2,1) \\
(1,3) \\
(1,2) \\
\end{array}
\right\},\left\{
\begin{array}{c}
 (2,3) \\
 (2,2) \\
(2,1) \\
(1,3) \\
(1,2) \\
(1,1) \\
\end{array}
\right\}$}.}
As for rating distributions, stochastic dominance on ${\mathcal X}$
can be characterized via increasing real-valued sequences, see e.g.\
Theorem 3.3.4 of \cite{Mueller2002}. More precisely, the relation
$\lambda\leq_{st}\mu$ holds for distributions on ${\mathcal X}$ iff
$\lambda x^T\leq \mu x^T$ for each increasing real-valued sequence
$x=(x_{a,r})_{(a,r)\in{\mathcal X}}$ in ${\mathbb R}^{KJ}$, where $x$
is called increasing if $(a,r)\leq (b,s)$ implies $x_{a,r}\leq
x_{b,s}$.

Another potential generalisation of the partial order on rating
distributions defined in (\ref{231028:1}) is formalised by the concept
of upper orthant orders. In our setup, the upper orthant order on
${\mathcal X}$ can be defined as follows. For $(a,r)\in{\mathcal X}$
let $U_{(a,r)}$ denote the upper set generated by $(a,r)$, i.e.,
$U_{(a,r)}=\{(b,s)\in{\mathcal X}\mid (a,r)\leq (b,s)\}$ and let
$\lambda=(\lambda_{a,r})_{(a,r)\in{\mathcal X}}$ and
$\mu=(\mu_{a,r})_{(a,r)\in{\mathcal X}}$ be distributions on
${\mathcal X}$. 
The {\em upper orthant order\/} $\leq_{uo}$ is defined by
$\lambda\leq_{uo}\mu$ if
\begin{equation} 
\sum_{(b,s)\in U_{(a,r)}} \lambda_{b,s}\leq  \sum_{(b,s)\in U_{(a,r)}}
\mu_{b,s}\ \ \ \ \ \hbox{for all } (a,r)\in{\mathcal
  X}. \label{231104:1} 
\end{equation}
Note that the definition of $\leq_{uo}$ is obtained by restricting the
upper sets in the definition (\ref{230327:2}) of $\leq_{st}$ to upper
sets that are generated by single elements $(a,r)\in{\mathcal
  X}$.\footnote{%
  From the list of upper sets in the previous footnote, this excludes
  the first, fifth, eigth and ninth sets. 
  } It
is an immediate consequence that $\lambda \leq_{st} \mu$ implies
$\lambda\leq_{uo} \mu$. Again, it is obvious that (\ref{231104:1}) is
equivalent to $\bar{F}_{\lambda}(a,r)\leq\bar{F}_{\mu}(a,r)$ for all
$(a,r)\in{\mathcal X}$, where $\bar{F}_{\lambda}(a,r)=\sum_{(a,r)\leq
  (b,s)}\lambda_{b,s}$ denotes the survival function of
$\lambda$.

In analogy to the ordering of ratings in a rating
transition matrix, a minimum requirement on the transition matrix
$P=(p_{(a,r), (b,s)})$ is
\begin{equation}
  (a,r)\leq (a',r') \text{ implies } p_{(a,r)} \leq_{uo}
  p_{(a',r')} \hbox{ for every } (a,r),(a',r')\in{\mathcal X},
  \label{eq:14}
\end{equation}
where the row vectors of $P$ are denoted by $p_{(a,r)}$.
In other words, $(a,r)\leq (a',r')$ implies $\p(A_{t+1}\geq b, R_{t+1}\geq
s|A_t=a, R_t=r)\leq \p(A_{t+1}\geq b, R_{t+1}\geq s|A_t=a', R_t=r')$, for all
$(b,s)\in \mathcal X$. 
It turns out, however, this is not 
sufficient to ensure that the order of multi-year PDs is preserved. 

To this end, we will now analyse the monotonicity of the transition matrix
$P
$ with respect to the
partial orders $\leq_{st}$ and $\leq_{uo}$. First of all, it can be
shown (e.g.\ Corollary 5.2.4 of \cite{Mueller2002}) that $P$ is
stochastically monotone, i.e., $\lambda\leq_{st} \mu$ implies $\lambda
P\leq_{st} \mu P$, if and only if the following generalisation of
(\ref{231028:3}) holds:
\begin{equation}
(a,r)\leq (a',r') \hbox{ implies } p_{(a,r)} \leq_{st} p_{(a',r')}
\hbox{ for every } (a,r),(a',r')\in{\mathcal X}.  
\label{230327:3}
\end{equation}
If $\leq_{st}$ is replaced by $\leq_{uo}$, i.e., \eqref{eq:14} is
considered instead of (\ref{230327:3}), a
condition for monotonicity with respect to $\leq_{uo}$ is obtained,
which is necessary but not sufficient. In order to ensure the
$\leq_{uo}$-monotonicity of $P$, further conditions on the row vectors
of $P$ are required, which can be derived from results in
\citet{Mueller2013}.

For $(a,r)\in{\mathcal X}$ define the distribution
$1^{(a,r)}=(1^{(a,r)}_{b,s})_{(b,s)\in{\mathcal X}}$ by
$1^{(a,r)}_{a,r}=1$ and $1^{(a,r)}_{b,s}=0$ for $(b,s)\not=
(a,r)$. Note that if $P$ is monotone with respect to $\leq_{st}$ or
$\leq_{uo}$ the matrix products $P^t$ satisfy 
\begin{equation}
(a,r)\leq (a',r') \hbox{ implies } 1^{(a,r)} P^t\leq_{uo} 1^{(a',r')} P^t,
\label{231102:1}
\end{equation}
for $(a,r),(a',r')\in{\mathcal X}$ and $t\in{\mathbb N}$.
An immediate consequence of condition (\ref{231102:1}) is that
multi-year default probabilities are ordered with respect to the
initial economic state and rating. More formally, we define the
multi-year default probabilities 
\[ {\rm PD}_t(a,r)=\p(R_{t}= \overline{r}\mid A_0=a,R_0=r), \ \ \
  t\in{\mathbb N}, \] 
for initial economic state and rating $(a,r)\in {\mathcal X}$. If the
transition matrix $P$ satisfies (\ref{231102:1}) then  
\begin{equation}
(a,r)\leq (a',r') \hbox{ implies } {\rm PD}_t(a,r)\leq {\rm
  PD}_t(a',r') \hbox{ for } t\in{\mathbb N}. 
\label{231104:2}
\end{equation}

For the rest of this section we will study monotonicity properties of
the transition matrix $M\in [0,1]^{K\times K}$ and the conditional
rating migration matrices $M^{(a,b)}\in [0,1]^{J\times J}$ that imply
stochastic monotonicity of the transition matrix $P\in [0,1]^{K
  J\times K J}$. It turns out that $M$ and the $M^{(a,b)}$ have to be
stochastic monotone and, in addition, the $M^{(a,b)}$ have to satisfy
a condition formulated in terms of the following order relation:
$M^{(a,b)}\leq_{st} M^{(a',b')}$ if $m_r^{(a,b)}\leq_{st}
m_r^{(a',b')}$ for all $r\in {\mathcal R}$, where $m_r^{(a,b)}$ and
$m_r^{(a',b')}$  denote the rows of $M^{(a,b)}$ and $M^{(a',b')}$
respectively.

\begin{theorem} ~\label{230104:t1}
Assume that the transition matrix $M\in [0,1]^{K\times K}$ and the
rating migration matrices $M^{(a,b)}\in [0,1]^{J\times J}$ have the
following properties:  
\begin{enumerate}[(i)]
  \item $M$ is stochastically monotone.
  \item $M^{(a,b)}$ is stochastically monotone for all
    $a,b\in{\mathcal A}$. 
\item $a\leq a'$ and $b\leq b'$ imply
  $M^{(a,b)}\leq_{st} M^{(a',b')}$ for all $a,a',b,b'\in{\mathcal
    A}$. 
  \end{enumerate}
Then the transition matrix $P\in [0,1]^{KJ\times KJ}$ is
stochastically monotone. 
\end{theorem}

\begin{proof}
Let $a,a'\in {\mathcal A}$ and $r,r'\in {\mathcal R}$ such that $a\leq
a'$ and $r\leq r'$. In order to prove (\ref{230327:3}) it suffices to
show that 
\begin{equation}
\sum_{(b,s)\in U} p_{(a,r),(b,s)}\leq \sum_{(b,s)\in U} p_{(a,r'),(b,s)}
\ \ \ \hbox{ and }\ \ \ 
\sum_{(b,s)\in U} p_{(a,r),(b,s)}\leq \sum_{(b,s)\in U} p_{(a',r),(b,s)},
\label{230117:1}
\end{equation}
where $U\subseteq {\mathcal X}$ is an arbitrary upper set with respect
to the order $\leq$ on ${\mathcal X}$.  

We define the partition 
\[ U=\bigcup_{b\in {\mathcal A}} \{(b,r)\mid r\in U_b\},\ \ \ U_b=\{r\in {\mathcal R}\mid (b,r)\in U\}, \]
where each $U_b$ is an upper set in $\mathcal R$. By Lemma
\ref{200125:l1}, $p_{(a,r),(b,s)}=m_{ab} m_{rs}^{(a,b)}$.  
Hence, the left inequality in (\ref{230117:1}) follows from the
assumption that the rating migration matrices 
$M^{(a,b)}$ are stochastically monotone for all $a,b\in{\mathcal A}$:
\[ \sum_{(b,s)\in U} p_{(a,r),(b,s)}
=\sum_{b\in{\mathcal A}}  m_{ab}\sum_{s\in U_b} m_{rs}^{(a,b)}
\leq \sum_{b\in{\mathcal A}}  m_{ab}\sum_{s\in U_b} m_{r's}^{(a,b)}
=\sum_{(b,s)\in U} p_{(a,r'),(b,s)}. \]

It remains to prove the right inequality in (\ref{230117:1}). 
By definition of the product order $\leq$ on ${\mathcal X}$, $b\leq c$
implies $U_b\subseteq U_c$ for all $b,c\in {\mathcal A}$. Together
with $M^{(a,b)}\leq_{st} M^{(a,c)}$ it follows that  
\[ \sum_{s\in U_b} m_{rs}^{(a,b)}\leq \sum_{s\in U_b} m_{rs}^{(a,c)}
\leq \sum_{s\in U_c} m_{rs}^{(a,c)}. \]
Hence, by defining $x_b=\sum_{s\in U_b} m_{rs}^{(a,b)}$ we obtain a
real-valued increasing sequence $(x_b)_{b\in{\mathcal A}}$, i.e.,
$b\leq c$ implies $x_b\leq x_c$. Since $M$ is stochastically monotone
it follows from (\ref{231028:2}) that  
\begin{equation}
\sum_{(b,s)\in U} p_{(a,r),(b,s)}=
\sum_{b\in{\mathcal A}}  m_{ab}\sum_{s\in U_b} m_{rs}^{(a,b)}
\leq
\sum_{b\in{\mathcal A}}  m_{a'b}\sum_{s\in U_b} m_{rs}^{(a,b)}.
\label{240103:1}
\end{equation}
Since $M^{(a,b)}\leq_{st} M^{(a',b)}$,
\[
\sum_{b\in{\mathcal A}}  m_{a'b}\sum_{s\in U_b} m_{rs}^{(a,b)}
\leq
\sum_{b\in{\mathcal A}}  m_{a'b}\sum_{s\in U_b} m_{rs}^{(a',b)}=
\sum_{(b,s)\in U} p_{(a',r),(b,s)}.
\]
Together with (\ref{240103:1}), we obtain the right inequality in
(\ref{230117:1}). 
\end{proof}

 \section{Ratings in the Merton model}
 \label{sec:stylised-example}

To demonstrate the above and further properties of PIT and TTC
ratings, we consider a Merton firm-value model \citep{Merton1974},
where an additional drift component, reflecting the economic state,
affects the firm's profitability.\footnote{For an application to real data
using historical S\&P rating data and 
GDP growth rates, see \citep{Gera2020}.}
More specifically, consider a firm's
{\em asset value process\/} ({\em credit quality process}, {\em
  ability-to-pay process}) $(V_t)_{t\geq0}$, with
\begin{equation}
  V_{t+1} = V_t \e^{(\mu + \mu_{t+1}^A-1/2 \sigma^2) +
    \sigma \Delta W_{t+1}}, 
  \label{eq:3}
\end{equation}
where, in the standard Merton-model, $\mu,\sigma>0$ are constants and
$W=(W_t)_{t\geq 0}$ is a Brownian motion (without loss of generality, we
set $\Delta t=1$). The economic state process $A_t$, independent of
$W$ introduces an additional drift in the credit quality process with
magnitude given by the one-to-one mapping $v:\mathcal A\rightarrow
\R$, i.e., the drift is given by $\mu_t^A=v(A_t)$.

The firm is in default at time $t+1$ if $\{V_{t+1}<D_{t+1}\}$, with
$D_{t+1}$ the $\mathcal F_t$-measurable default threshold. In a Markov
setting, $D_{t+1}$ will be $\sigma(V_t, \mu_t^A)$-measurable. If one
associates with $V$ the firm's asset value, then $D$ corresponds to
the firm's debt value, with default taking place if the asset value
drops below the debt value. We give details on how the firm updates
its debt level $D$ in Assumption \ref{assumption:pit} below.

To ease notation, we write the log asset-to-debt ratio process
$Z_t=\ln(V_t/D_t)$. Because default takes place if $Z_t<0$, this
process has an interpretation as the ``distance to default''. Because
the firm updates $Z_t$ (e.g.\ by updating $D_t$), we also write
$Z_{(t+1)-}=\ln(V_{t+1}/D_t)$ to denote the random variable, which
initially at time $t$ reflects the ``distance to default'', and in
response to which the firm updates $D_t$ to $D_{t+1}$.

\subsection{Point-in-time rating}
\label{sec:pit-rating}

At any time, a firm's PIT rating is determined by its probability of
default. The non-default rating classes $r_1, \ldots, r_{J-1}$
are specified via PD buckets $(d_{j-1}, d_j]$, $j=1, \ldots J-1$,
with PD boundaries $\{d_0=0, d_1, \ldots, d_{J-1}=1\}$. In
other words, rating $r_j$ is assigned at time $t$ if
the firm's probability of default at time $t$ is in $(d_{j-1},
d_j]$. To ensure that the 
overall Markov process is time-homogeneous, we assume that the firm
makes adjustments to its capital structure over time. This assumption
is justified e.g.\ by \citep{Lemmon2008} who observe that capital
structures and leverage ratios tend to remain stable over time.
In the following, we assume that the firm is not in default at time
$t$.
\begin{assumption}
 \label{assumption:pit}
 At time $t$, the initial debt level is $D_t$. Denote
  $Z_{(t+1)-} = \ln(V_{t+1}/D_t)$. The rating class $r_j$ at time $t$
  is determined from
  $\p(Z_{(t+1)-}< 0|Z_t, \mu_t^A)\in (d_{j-1}, d_j]$. 
  Associated
  with each rating class is a PD $p_j\in (d_{j-1}, d_j]$,
  $j=1,\ldots, J$. Based on the rating class $r_j$, the firm chooses
  its debt level $D_{t+1}$ such that $\p(Z_{t+1}<0|Z_t, \mu_t^A)=p_j$.
\end{assumption}

By construction, a PIT rating is fully determined by a firm's PD, so a
firm with rating $r_j$ has 
a probability of default of $p_j$, regardless of
the state of the economy. Formally,
\begin{equation*}
  R_t=r_j \quad \iff \quad \p(Z_{(t+1)-}<0|Z_t, \mu_t^A)\in
  (d_{j-1},d_j] \quad\iff\quad \p(Z_{t+1}<0|Z_t, \mu_t^A)=p_j.
\end{equation*}
This is how PIT ratings are {\em assigned}. The following proposition
gives details about the rating process obtained by this construction.

\begin{proposition}~
  \begin{enumerate}[(i)]
  \item At time $t$, conditional on no-default, i.e., $R_t\in
    \widehat{\mathcal R}$, which is equivalent to $\{Z_t\geq 0\}$, the
    firm's initial probability of default is given by
    \begin{equation*}
      \p(Z_{(t+1)-}<0|Z_t, \mu_t^A)
      = \sum_{a\in \mathcal A} \Ncdf\left(-\frac{Z_t  + \mu
          +  v(a) - \frac{1}{2}\sigma^2 }
        {\sigma}\right) \p(\mu_{t+1}^A=v(a)|\mu_t^A),
    \end{equation*}
    where $\Ncdf$ denotes the standard normal distribution function. 
  \item For each economic state, there exist asset-to-debt 
    ratios $\alpha_1, \ldots, \alpha_{J-1}$, determined by solving,
    for each rating class $r_j$,
    \begin{multline}
      \label{eq:7}
      \p(R_{t+1}=\overline r|R_t=r_j,A_t) =p_j
      = \sum_{a\in \mathcal A}
      \Ncdf\left(-\frac{\ln(\alpha_j) + \mu + v(a) -
          \frac{1}{2}\sigma^2}{\sigma}\right) \p(\mu_{t+1}^A=v(a)|\mu_t^A),
    \end{multline}
    where $p_j$ is the default probability associated
    with rating $r_j$. The asset-to-debt ratio correspond to
    $\alpha_j=V_t/D_{t+1}$, $j=1,\ldots, J-1$.  In
    particular, the left equation implies that the rating is a PIT
    rating, see Definition \ref{def:PIT}.
  \item There exist rating boundaries
    $c_{J+1}=-\infty < c_J=0 < c_{J-1} < \cdots < c_1=\infty$
    depending on $A_t$, $A_{t+1}$ and $R_t$ such
    that the conditional rating transition probabilities $m_{r\, 
      r_j}^{(a,b)}$ are given 
    by
    \begin{equation*}
      m_{r\, r_j}^{(a,b)} = \p(R_{t+1}=r_j |R_t=r, A_t=a, A_{t+1}=b
      ) =
      \Ncdf(c_j)-\Ncdf(c_{j+1}),\quad j=1,\ldots, J.
    \end{equation*}
    Using that the asset-to-debt ratio is $\alpha$ given $R_t=r$ and
    $A_t=a$, the boundaries are
        \begin{align}
      c_{j+1}&= \frac{\min(\tilde c_{j+1},0) -\ln(\alpha)-
               \mu-v(b) + 1/2\sigma^2}{\sigma}\label{eq:8}\\
      c_j &= \frac{\tilde c_j -\ln(\alpha) - \mu-v(b) +
            1/2\sigma^2}{\sigma},\label{eq:12}
    \end{align}
    where $\tilde c_j, \tilde c_{j+1}$, depending on $A_t=a$, are such
    that
    \begin{equation*}
      \p(Z_{(t+1)-} <0 | Z_{t},\mu_{t}^A=v(a))\in (d_{j-1}, d_j] \quad
      \iff\quad Z_{t} \in [\tilde c_{j+1}, \tilde c_j). 
    \end{equation*}
  \item The process $(R_t, A_t
    )$ is a time-homogeneous Markov
    process.
  \end{enumerate}
\end{proposition}

\begin{proof}
 (i) The claim is obtained from
    $Z_{(t+1)-} = Z_t + \mu + \mu_{t+1}^A -
    \frac{1}{2}\sigma^2 + \sigma\Delta W_{t+1}$ and taking expectation
    with respect to $\mu_{t+1}^A$.

\medskip
\noindent
(ii) By construction $R_t=r_j$ if and only if 
    $\p(Z_{t+1}<0|Z_t,\mu_t^A)=p_j$.
    Since the asset-to-debt ratio $V_t/D_{t+1}$ is
    $(Z_t,\mu_t^A)$-measurable, $V_t/D_{t+1}$ may be replaced by 
    $\alpha_j$, so that
    \begin{align*}
      Z_{t+1} &= \ln(V_t/D_{t+1}) + \mu + \mu_{t+1}^A - 1/2\sigma^2 +
                \sigma\Delta W_{t+1}\\
              &= \ln(\alpha_j) + \mu + \mu_{t+1}^A - 1/2\sigma^2 +
                \sigma\Delta W_{t+1}.
    \end{align*}
    The left equation of \eqref{eq:7} holds by
    definition of the rating classes, while the right equation holds
    by construction of the PIT rating system on the firm-value process.
    It remains to
    observe that the sum in \eqref{eq:7} is a mixture of functions
    strictly monotone in $\alpha_j$; hence, $\alpha_j$ is
    uniquely determined given $p_j$. 
    \medskip
    \noindent
    (iii)
If $j=J$, then
    \begin{equation*}
      \p(R_{t+1}=\overline{r}|R_t, \mu_t^A, \mu_{t+1}^A) = \p(Z_{(t+1)-}<0 |R_t,
      \mu_t^A, \mu_{t+1}^A), 
    \end{equation*}
    so $c_{J+1}=-\infty$ and $c_J=0$.  In the case $j\not=J$, we have
    \begin{multline}
      \label{eq:5}
      \p(R_{t+1}=r_j|R_t, A_t, A_{t+1}
      )\\ %
      \begin{aligned}
        &= \p\left(\p(Z_{t+2}<0|Z_{t+1},\mu_{t+1}^A)=p_j,
          Z_{t+1}>0\big| Z_t, \mu_t^A, \mu_{t+1}^A\right) \\ %
        &= \p\left(\p(Z_{(t+2)-} <0 | Z_{t+1},\mu_{t+1}^A)\in
          (d_{j-1}, d_j], Z_{t+1}>0\big| Z_t, \mu_t^A,
          \mu_{t+1}^A\right).
      \end{aligned}
    \end{multline}
    On $\{Z_{t+1}\geq 0\}$, holding $\mu_{t+1}^A$ fixed, the PD
    $\p(Z_{(t+2)-}<0 |Z_{t+1}, \mu_{t+1}^A)$ is invertible in
    $Z_{t+1}$, since from (i):
    \begin{equation}
      \label{eq:5}
      \p(Z_{(t+2)-}<0 |Z_{t+1}, \mu_{t+1}^A) %
      = \sum_{k=1}^K \Ncdf\left(-\frac{Z_{t+1} + \mu + \mu_{t+2}^A -
          \frac{1}{2}\sigma^2} {\sigma}\right)\,
      \p(\mu_{t+2}^A=v(a_k)|\mu_{t+1}^A),  
    \end{equation}
    which is a mixture of functions strictly monotone in
    $Z_{t+1}$. Hence, there exist boundaries
    $\tilde c_{j+1}, \tilde c_j$ dependent on $\mu_{t+1}^A$ such that,
    \begin{equation*}
      \p(Z_{(t+2)-} <0 | Z_{t+1},\mu_{t+1}^A)\in (d_{j-1}, d_j] \quad
      \iff\quad Z_{t+1} \in [\tilde c_{j+1}, \tilde c_j). 
    \end{equation*}
    Inserting in \eqref{eq:5} and using that, if $R_t=r$,
    \begin{equation*}
      Z_{t+1} =\ln(\alpha) + \mu + \mu_{t+1}^A - 1/2\sigma^2 +
      \sigma\Delta W_{t+1}, 
    \end{equation*}
    gives 
    \begin{multline*}
      \p(R_{t+1}=r_j|R_t=r, A_t=a, A_{t+1}=b) \\%
      \begin{aligned}
        &= \p(Z_{t+1}\in [\min(\tilde c_{j+1},0), \tilde
        c_j)|V_t/D_{t+1}=\alpha, \mu_t^A=v(a), \mu_{t+1}^A=v(b)) \\
        &= \p(\Delta W_{t+1} \in [c_{j+1}, c_j)),
      \end{aligned}
    \end{multline*}
    where $c_{j+1}$ and $c_j$ are given by \eqref{eq:8} and
    \eqref{eq:12}, respectively. 
 \medskip
\noindent
(iv) It is a consequence of Assumption
    \ref{assumption:pit} and parts (ii), (iii) that the rating
    transition probabilities do not depend on time or on the history
    of the asset-to-debt ratio. 
\end{proof}

\subsection{Through-the-cycle rating}
\label{sec:through-cycle-rating-1}

The PIT rating construction implies that, when observing the
economic state, it is equivalent to observe a firm's credit
quality or rating process. The Markov property is secured
via Assumption \ref{assumption:pit}, which links the credit quality 
(via the asset-to-debt ratio) to the firm's rating class.

Conditional TTC rating changes are determined independent of the
economic state, for example due to credit quality
changes idiosyncratic to the firm. As a consequence, building a TTC
rating on the process \eqref{eq:3} following Assumption 
\ref{assumption:pit} will, in general, be non-Markovian, as observing
the history of economic state transitions will reveal additional
information on a firm's probability of default.
Hence, a similar construction for a TTC rating requires different
assumptions on the firm's activities: achieving the Markov property 
requires the firm to eliminate credit quality information from
systematic changes (i.e., the economic state) at each point in time.
This stronger assumption will, in general, 
lead to a time-inhomogeneous Markovian PIT rating. 
Depending on the economic state, one would have to adjust rating
boundaries to yield probabilities of default that are independent of
the economic state.  

Not only is the stronger assumption suggested above economically
implausible, but from a practical point of view, it is not desirable to
have two disconnected rating systems yielding entirely different
rating classes. Rather, it is reasonable to build on the PIT property, 
i.e., the connection between rating class and probability of
default, while enforcing the TTC property. 

As a consequence, we shall stick to the construction of the
PIT rating, in particular Assumption \ref{assumption:pit}, and
establish the closest TTC rating by minimising the discrepancy between
TTC and PIT default term structures, conditional on economic states
and PIT ratings. This provides a feasible construction of a TTC rating
from a given PIT rating, even without the explicit PIT rating
construction from the previous section. 

In the following, we assume throughout that the number of rating
classes are the same for both PIT and TTC ratings. 

\begin{definition}
  A firm's {\em default term structure\/} at time $0$ 
  is defined as
  \begin{equation*}
    t\mapsto p_t:=
    \p(R_t=\overline r|R_{t-1}\not= \overline r), \quad t>0. 
  \end{equation*}
\end{definition}
PDs entering the default term structure are
calculated as follows:  with $\lambda_0$ the initial
distribution of rating and economic state, $\lambda_t=\lambda_0 \cdot
P^t$ is the time-$t$ distribution, where $P$ is the rating
transition matrix given by \eqref{eq:4}. 
The time-$t$ default probability is given by $\p(R_t=\overline
r)=\sum_{a\in \mathcal A} \lambda_{t,(a,\overline r)}$, where
$\lambda_{t,(a,r)}$, $a\in \mathcal A$, $r\in \mathcal R$, denote the
entries of $\lambda_t$.  

The construction of the TTC rating could be done in one step, but due
to the computational burden, we break it down into two steps:
\begin{enumerate}[(I)]
\item Define the TTC $Q$ matrix (see Definition \ref{def:QTTC}), which
  embodies the non-default rating transitions and is independent of
  the economic state.           
\item Determine the TTC $P$ matrix and TTC initial distributions
  conditional on economic state and PIT rating. 
\end{enumerate}
The two steps are now described in detail: \\

(I). There are different plausible ways to build a TTC's rating $Q$
matrix: 
one can average the PIT rating's $Q^{(a,b)}$ matrices, which depend on
the economic state transitions. The averaging may be weighed according
to the probabilities of economic state transitions:
\begin{equation*}
  \p( A_{t+1}=b, A_t=a) = \p(A_{t+1}=b | A_t=a) \, \p(A_t=a) = m_{ab}
  \cdot \mu_a^A,\quad a,b\in \mathcal A,
\end{equation*}
where $\mu^A$ denotes the stationary distribution of the economic
states.

Another method -- the one used in the numerical example below -- would
use the PIT ratings' asymptotic $Q$ matrix, which is
denoted by $\tilde Q$ and derived from $\tilde P_\infty$, see Equations
\eqref{eq:2} and \eqref{240227:2}. 
\medskip

(II). From Equation \eqref{eq:10},
a TTC rating transition matrix is characterised by 
\begin{equation}
M^{(a,b)} = D^{(a,b)}\cdot Q, \quad a, b\in \mathcal A. 
\label{eq:6}
\end{equation}
It remains to specify the default 
probabilities $m_{r_j\overline r}^{(a,b)}$, $j=1, \ldots, J-1$, of the
default component matrix $D^{(a,b)}$.
The goal is to determine the default 
probabilities by calibrating default-term structures for each
combination of PIT rating and economic state. This requires
calibrating probabilities of the initial TTC rating for each of these
combinations as well. More specifically, for each initial state
$(a,r)$, the initial probabilities under the TTC rating,
$\lambda_{(a,r)}^{\text{TTC}}\in [0,1]^{KJ}$, are determined by
calibrating the default term structure conditional on the economic
state $a$ and the PIT rating $r$. In addition, 
default-term structures are fitted where only the economic state is
given and the initial PIT rating is specified by the stationary
distribution of the non-default states. The corresponding
TTC default-term structures (for each economic state) are also
determined from the TTC rating's stationary distributions.
Formally, the calibration problem is specified as 
\begin{equation}
  \label{eq:1}
  \min_{\{m_{r\overline r}^{(a,b)},
    \lambda_{(a,r)}^{\text{TTC}}| r\in
    \widehat{\mathcal R}, a,b\in \mathcal A\}} %
  \left(\sum_{r\in \widehat{\mathcal R}}\sum_{a\in \mathcal A}\sum_{t\in \mathcal T}
    \left(p_{t,(a,r)}^{\text{TTC}}-p_{t,(a,r)}^{\text{PIT}}\right)^2 +
    \sum_{a\in \mathcal A} \sum_{t\in \mathcal T} (p_{t,a}^{\text{TTC}}
    - p_{t,a}^{\text{PIT}})^2\right)^{1/2}. 
\end{equation}
As outline above, $(p_{t,(a,r)}^{\text{PIT}})_{t\in \mathcal T}$
and $(p_{t,(a,r)}^{\text{TTC}})_{t\in \mathcal T}$ as well as
$\lambda_{(a,r)}^{\text{TTC}}$ refer to default
term structures, resp.\ an initial distribution, conditional on
economic state $a$ and PIT rating $r$, whereas $(p_{t,a}^{\text{TTC}})_{t\in
  \mathcal T}$ and $(p_{t,a}^{\text{PIT}})_{t\in \mathcal T}$ refer to
default term structures where the initial rating distribution is
determined from the respective stationary distribution conditional on
economic state $a\in \mathcal A$. 
Constraints of the optimisation problem ensure that the initial
distribution is a probability measure, that the matrix
$P^{\text{TTC}}$ is a stochastic matrix and that PDs of default term
structures with varying initial distributions are well-ordered.

\subsection{Example}
\label{sec:example}

\begin{table}[h]
  \centering
  \begin{tabular}{cc}
    Parameter& Value\\\hline
    \multicolumn{2}{c}{Merton model}\\\hline
    $\mu$ & 0.003\\
    $\sigma$ & 0.006\\
    $\Delta t$ & 1\\\hline
    \multicolumn{2}{c}{Economic states}\\\hline
    \# economic states & 3\\
    Economic state transition matrix $M$ &
                                       $\begin{pmatrix}
                                         0.8 & 0.175 & 0.025\\
                                         0.1 & 0.8 & 0.1\\
                                         0.025 & 0.175 & 0.8
                                       \end{pmatrix}$\\
    Economic state drift $\mu^A$& $\{\sigma, 0, -\sigma\}$\\\hline
    \multicolumn{2}{c}{Rating classes}\\\hline
    \# non-default rating classes & 3\\
    PIT default probabilities & $p\in \{0.0002, 0.005, 0.025\}$\\
    PIT rating boundaries & $d\in \{0, 0.0003, 0.02, 1\}$\\\hline
  \end{tabular}
  \caption{Parameters used in the example.}
  \label{tab:parameters}
\end{table}
This section illustrates PIT and TTC ratings determined from a Merton
model as outlined in the previous sections. The parameters used are
shown in Table \ref{tab:parameters}.
\begin{figure}[h]
  \centering
  \includegraphics[width=.45\textwidth]{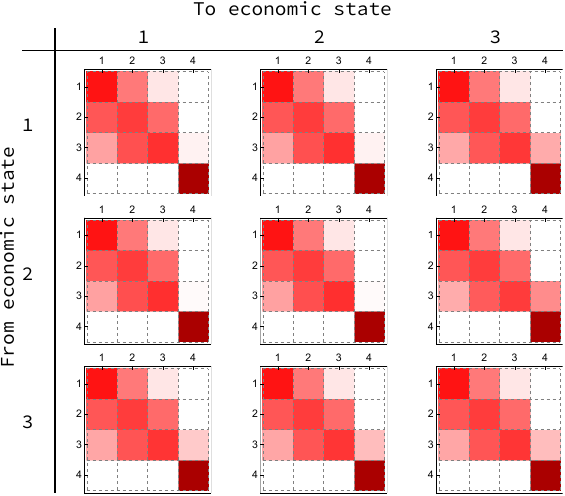}\ \ \
  \ \ \
  \includegraphics[width=0.4\textwidth]{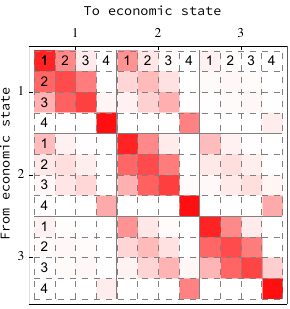} \ \ \ 
  \caption{$M^{(a,b)}$ matrices (left) and $P$ matrix
    (right) of TTC rating. 
    Economic states are labelled 1 (good), 2 (neutral) and 3
    (bad). Darker shades indicate values closer to one, whereas white
    indicates a value of zero or close to zero.}
  \label{fig:rating_example_ttc}
\end{figure}

Figure \ref{fig:rating_example} in Section
\ref{sec:example-rating-syst} shows the PIT $M^{(a,b)}$ matrices and
$P$ matrix. Figure \ref{fig:rating_example_ttc} shows the
corresponding matrices of the TTC rating.  
The $M^{(a,b)}$ matrices reveal that, as expected, rating transitions
in the PIT rating depend more strongly on the economic state
transition, while transitions to non-default ratings of the TTC system
vary little across economic state transitions. The PDs of TTC rating
classes (last column in each matrix) show variation across the
economic states. 
It is easily verified that the $P$
matrices are monotone with respect to the upper orthant order. They
are even stochastically monotone, so powers are monotone with respect
to the upper orthant order as well.

\begin{figure}[t]
  \centering
  \includegraphics[width=0.75\textwidth]{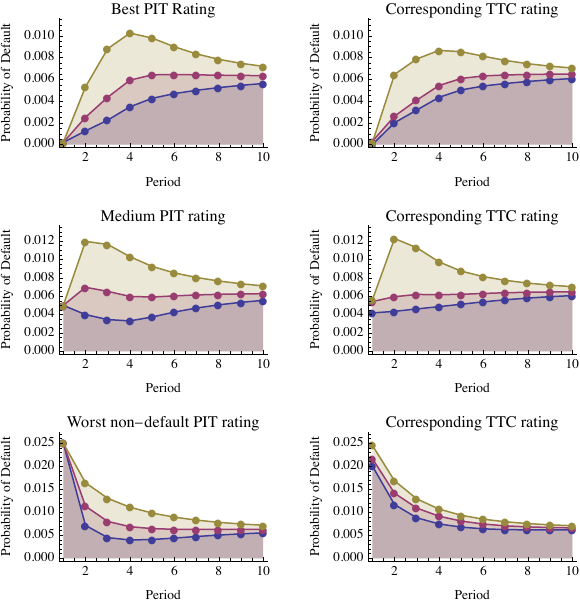}
  \caption{Default curves generated by PIT (left) and TTC (right)
    ratings. Left, from top to bottom: initial PIT ratings are ordered
    from best to worst non-default rating; each curve denotes the
    economic state in which the curve starts (blue: good, red: neutral,
    green: bad). Right: Corresponding TTC ratings, where initial
    distributions $\lambda^{\text{TTC}}_{(a,r)}, r\in \mathcal{\hat R},
    a\in \mathcal A$, are calibrated to respective initial PIT rating,
    see Equation (\ref{eq:1}).
  }
  \label{fig:defaultCurves1}
\end{figure}
In the following, we compare various default curves generated by PIT
and TTC ratings. The left-hand side of Figure \ref{fig:defaultCurves1}
shows default curves for all combinations of initial PIT non-default
ratings and economic states. The right-hand side shows the
corresponding calibrated TTC ratings. Regardless of the economic
state, the PIT one-period PD depends only on the initial PIT
rating. This is not the case for TTC ratings, although default
probabilities are close, given that each TTC default curve represents
firms in one PIT rating. The default curves converge at the long end
across ratings and economic states. Because TTC rating classes are
more stable, due to the constant $Q$ matrix, convergence tends to
occur faster. 

\begin{figure}[t]
  \centering
  \includegraphics[width=\textwidth]{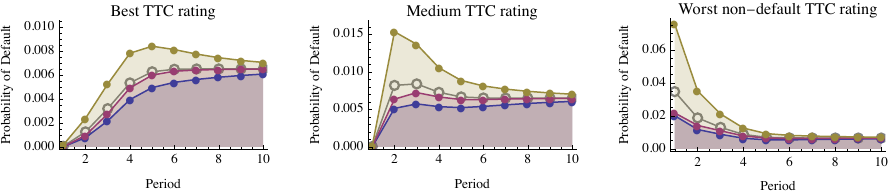}
  \caption{Default curves generated by TTC ratings. Colour-coded
    curves denote the initial economic
state (blue: good, red: neutral, green: bad). The grey curve denotes
the default probabilities of the asymptotic approximation
\eqref{eq:2}, which corresponds to a TTC rating that has no
economic-state information at all.}
  \label{fig:defaultCurves2}
\end{figure}
Figure \ref{fig:defaultCurves2} shows default curves for every
combination of initial TTC rating and initial economic state. In
each graph, the curves are again ordered according to the initial
economic state (blue: good, red: neutral, green: bad). As the
one-period PDs for the best and medium TTC ratings are practically
zero, regardless of the economic state, a firm's rating path will
nearly certainly pass through the worst TTC rating before it
defaults.

Also shown in Figure \ref{fig:defaultCurves2} (in grey) are default
curves derived from the asymptotic approximation \eqref{eq:2} of the
TTC rating. As before, (conditional) rating migrations are insensitive
to the economic state, but in addition, this is now also the case for 
PDs. These PDs meet regulatory capital requirements by being 
based on the observed historical average one-year default rate,
\citet{BIS2016a}. It can be seen that these PDs over- or underestimate
real PDs depending on the economic state. 

\begin{figure}[t]
  \centering
  \includegraphics[width=.75\textwidth]{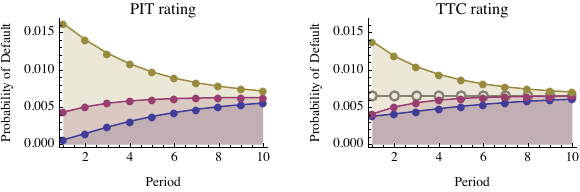}
  \caption{Default curves generated by PIT and TTC ratings where only
    the initial economic state is known. The initial distributions of
    the PIT and TTC ratings are determined from the stationary
    distribution of $P$. Colour-coded curves denote the initial
    economic state (blue: good, red: neutral, green: bad). The grey
    curve denotes the default probabilities of the asymptotic
    approximation \eqref{eq:2}, which corresponds to a TTC rating that
    has no economic-state information at all.}
  \label{fig:defaultCurves3}
\end{figure}
Finally, Figure \ref{fig:defaultCurves3} shows default curves of PIT
and TTC ratings, where only the initial economic state is known. The
initial distributions are chosen as the stationary distributions,
which are determined from the $P$ matrices with default states
removed, and then conditioned on each economic state. The curves can
be interpreted as representing -- conditional on the economic state --
the default probabilities of a randomly picked firm. As such, ideally,
the curves would be equal, as the default probabilities of a randomly
picked firm should not depend on a particular rating philosophy.
However, as can be seen, with the TTC rating being ``more stable'', it
is harder to differentiate the neutral and bad economic states than in
the case of the PIT ratings where the initial default probabilities
are given.

The grey curve in Figure \ref{fig:defaultCurves3} shows the default
curve of the TTC rating's asymptotic approximation \eqref{eq:2}. Given
that there is no initial information about the firm and the economic
state, these PDs remain constant through time. 

\section{Conclusion}
\label{sec:conclusion}
To conclude, we point out how the rating framework introduced in this
paper complies with requirements of ratings and PDs
in both accounting and regulatory capital standards. 

The International Financial Reporting Standard
(IFRS) 9 as well as the accounting standard Current Expected Credit
Losses (CECL) require financial institutions to calculate expected
credit losses that reflect current economic conditions and forecasts
of future economic conditions\footnote{International Financial
  Reporting Standard (IFRS) 9, Paragraph 5.5.17}. Hence, default
probabilities that are used for calculating expected credit losses for
accounting purposes have to be realistic estimates of the
actual probability of default taking the current economic state as
well as the future (stochastic) development of economic conditions
into account. As the accounting
standards do not contain any requirements on rating systems, even
default probabilities of a rating process 
$(R_t)_{t\in \Nzero}$ that is predominantly TTC satisfies the accounting
requirements as long as the PD curve ${\p(R_t=\overline{r})}_{t\in \Nzero}$
specified by $(R_t)_{t\in \Nzero}$
provides a realistic estimate of the actual probabilities of
default. This requirement typically implies that default probabilities 
$\p(R_{t+1}=\overline{r}|R_t=r)$ show a strong dependence on economic
states to make up for the low sensitivity of non-default ratings to
economic conditions if $(R_t)_{t\in \Nzero}$ is predominantly TTC.

The PD curve ${\p(R_t=\overline{r})}_{t\in \Nzero}$ can be derived
from the time-inhomogeneous transition matrices ${\widetilde P}_t$ via
Equation (\ref{240210:1}). In our setup, the transition matrices
${\widetilde P}_t$ depend on the probabilities of the different
economic paths $(a_0,a_1,\ldots, a_t)\in{\mathcal A}^{t+1}$ specified
by the Markov process $(A_t)_{t\in \Nzero}$. In practice, the
specification of economic scenarios $(a_0,a_1,\ldots, a_T)$ is based
on actual macro-economic forecasts, which provide the economic
information required for the transition matrices  ${\widetilde
  P}_1,\ldots,{\widetilde P}_T$ up to a forecast horizon $T$. For
extending the default curve beyond $T$, the asymptotic transition
matrix ${\widetilde P}_{\infty}$ is typically used, which formalizes
the long-term average migrations of the rating process, see Equation
\eqref{eq:2}.
\enlargethispage{\baselineskip}

In contrast to accounting requirements,
financial regulators wish to avoid procyclical capitalisation of financial
institutions, as this would incentivise firms to reduce their capital
when the economy thrives, which in turn would require firms to increase
their capital in economic downturns, when this is particularly difficult
to achieve. 
As this procyclical capitalisation
poses a threat to financial stability, financial regulators
prefer banks to assess regulatory capital requirements for credit risk
using through-the-cycle (TTC) 
ratings,\footnote{Capital requirements regulation and directive --
CRR/CRD IV, Article 502} which are designed to be stable through the
business cycle.
In addition, default probabilities in regulatory capital requirement
calculations should be based on the observed 
historical average one-year default rate, see \citet{BIS2016a}. In
order to satisfy these requirements in our setup, consider a rating
migration process $(R_t)_{t\in \Nzero}$ that is predominantly TTC.
The asymptotic PDs of $(R_t)_{t\in \Nzero}$ correspond to the historical
average one-year default rate of this rating migration process.
Hence, 
the asymptotic approximation $(S_t)_{t\in \Nzero}$ based on the
transition matrix ${\widetilde P}_{\infty}$ (see Definition 
\ref{230102:d1}) is a rating
migration process that is consistent with the regulatory requirements
listed above. However, as the example in Section \ref{sec:example}
demonstrates, $(S_t)_{t\in \Nzero}$ does
not provide realistic PD estimates. In particular, these PDs -- shown
in Figures \ref{fig:defaultCurves2} and \ref{fig:defaultCurves3} -- 
tend to underestimate the actual default risk in a negative economic
environment. Alternatively, the asymptotic default probabilities specified by
${\widetilde P}_{\infty}$ could be replaced by default probabilities
taken from a conditional transition matrix $M^{(a,b)}$, where the
choice of the economic states $a,b\in {\mathcal A}$ reflects the
underlying economic assumptions. For example, if the transition from
$a$ to $b$ represents a negative economic scenario the default
probabilities in $M^{(a,b)}$ could be interpreted as stressed PDs. If
default probabilities are specified in this way the underestimation of
the actual default risk could be avoided, even in a negative economic
environment, but at the expense of overly conservative PD estimates
during benign economic times.

\ifaap
\else
\setlength{\bibsep}{4pt plus 0.3ex}
\bibliographystyle{abbrvnamed} %
\fi
\bibliography{finance} %
\end{document}

%% file: definitions.tex
\usepackage{mathrsfs}

\providecommand{\R}{{\mathbb{R}}}

\newcommand{\p}{{\mathbb P}}

\newcommand{\e}{{\bf e}}

\providecommand{\N}{{\mathbb N}}
\providecommand{\Nzero}{\mathbb N_0}
\providecommand{\Ncdf}{{\rm N}}

\ifx\prop\undefined
\newtheorem{prop}{Proposition}[section]
\fi
\ifx\theo\undefined
\newtheorem{theo}[prop]{Theorem}
\fi
\ifx\lem\undefined
\newtheorem{lem}[prop]{Lemma}
\fi
\ifx\cor\undefined
\newtheorem{cor}[prop]{Corollary}
\fi
\ifx\defi\undefined
\newtheorem{defi}[prop]{Definition}
\fi












%% file: ratings.bbl
\begin{thebibliography}{}

\bibitem[\protect\citeauthoryear{Aguais \bgroup \em et al.\egroup
  }{2008}]{Aguais2008}
Aguais, S.~D., L.~R. Forest, M.~King, M.~C. Lennon, and B.~Lordkipanidze.
\newblock Designing and {I}mplementing a {Basel II Compliant PIT-TTC Ratings
  Framework}.
\newblock In {\em The Basel Handbook}. Risk Books, 2nd edition, 2008.

\bibitem[\protect\citeauthoryear{Altman and Rijken}{2004}]{Altman2004}
Altman, E.~I. and H.~A. Rijken.
\newblock How rating agencies achieve rating stability.
\newblock {\em Journal of Banking \& Finance}, 28(11):2679--2714, 2004.

\bibitem[\protect\citeauthoryear{Altman}{1998}]{Altman1998}
Altman, E.~I.
\newblock The importance and subtlety of credit rating migration.
\newblock {\em Journal of Banking \& Finance}, 22(10-11):1231--1247, 1998.

\bibitem[\protect\citeauthoryear{Bangia \bgroup \em et al.\egroup
  }{2002}]{Bangia2002}
Bangia, A., F.~X. Diebold, A.~Kronimus, C.~Schagen, and T.~Schuermann.
\newblock Ratings migration and the business cycle, with application to credit
  portfolio stress testing.
\newblock {\em Journal of Banking \& Finance}, 26:445--474, 2002.

\bibitem[\protect\citeauthoryear{{Basel Committee on Banking
  Supervision}}{2016}]{BIS2016a}
{Basel Committee on Banking Supervision}.
\newblock Reducing variation in credit risk-weighted assets - constraints on
  the use of internal model approaches.
\newblock Consultative Document, 2016.

\bibitem[\protect\citeauthoryear{Behn \bgroup \em et al.\egroup
  }{2016}]{Behn2016}
Behn, M., R.~Haselmann, and P.~Wachtel.
\newblock Procyclical capital regulation and lending.
\newblock {\em The Journal of Finance}, 71(2):919--956, 2016.

\bibitem[\protect\citeauthoryear{Bluhm and Overbeck}{2007}]{Bluhm2007}
Bluhm, C. and L.~Overbeck.
\newblock Calibration of {PD} term structures: to be {Markov} or not to be.
\newblock {\em Risk}, 20(11):98--103, 2007.

\bibitem[\protect\citeauthoryear{Bluhm \bgroup \em et al.\egroup
  }{2003}]{Bluhm2003}
Bluhm, C., L.~Overbeck, and C.~Wagner.
\newblock {\em An Introduction to Credit Risk Modeling}.
\newblock Chapman \& Hall/CRC, London, 2003.

\bibitem[\protect\citeauthoryear{Borio \bgroup \em et al.\egroup
  }{2001}]{Borio2001}
Borio, C., C.~Furfine, and P.~Lowe.
\newblock Procyclicality of the financial system and financial stability:
  issues and policy options.
\newblock BIS papers, No.\ 1, 2001.

\bibitem[\protect\citeauthoryear{Br{\'e}maud}{1999}]{Bremaud1999}
Br{\'e}maud, P.
\newblock {\em Markov chains: Gibbs fields, Monte Carlo simulation, and
  Queues}, volume~31.
\newblock Springer, 1999.

\bibitem[\protect\citeauthoryear{Carlehed and Petrov}{2012}]{Carlehed2012}
Carlehed, M. and A.~Petrov.
\newblock A methodology for point-in-time-through-the-cycle probability of
  default decomposition in risk classification systems.
\newblock {\em The Journal of Risk Model Validation}, 6(3):3, 2012.

\bibitem[\protect\citeauthoryear{Crouhy \bgroup \em et al.\egroup
  }{2001}]{Crouhy2001}
Crouhy, M., D.~Galai, and R.~Mark.
\newblock Prototype risk rating system.
\newblock {\em Journal of Banking \& Finance}, 25(1):47--95, 2001.

\bibitem[\protect\citeauthoryear{Duffie and Singleton}{2003}]{Duffie2003}
Duffie, D. and K.~J. Singleton.
\newblock {\em Credit Risk: Pricing, Measurement, and Management}.
\newblock Princeton University Press, 2003.

\bibitem[\protect\citeauthoryear{Fei \bgroup \em et al.\egroup
  }{2012}]{Fei2012}
Fei, F., A.-M. Fuertes, and E.~Kalotychou.
\newblock Credit rating migration risk and business cycles.
\newblock {\em Journal of Business Finance \& Accounting}, 39(1-2):229--263,
  2012.

\bibitem[\protect\citeauthoryear{F\"ollmer and Schied}{2002}]{Foellmer2002}
F\"ollmer, H. and A.~Schied.
\newblock {\em Stochastic Finance. An Introduction in Discrete Time}.
\newblock de Gruyter, 2002.

\bibitem[\protect\citeauthoryear{Frobenius}{1912}]{Frobenius1912}
Frobenius, G.
\newblock {\"U}ber {M}atrizen aus nicht negativen {E}lementen.
\newblock {\em Sitzungsberichte der K\"oniglich Preussischen Akademie der
  Wissenschaften}, pages 456--477, 1912.

\bibitem[\protect\citeauthoryear{Frydman and Schuermann}{2008}]{Frydman2008}
Frydman, H. and T.~Schuermann.
\newblock Credit rating dynamics and markov mixture models.
\newblock {\em Journal of Banking \& Finance}, 32(6):1062--1075, 2008.

\bibitem[\protect\citeauthoryear{Gera}{2020}]{Gera2020}
Gera, I.
\newblock Credit rating migration processes based on economic state-dependent
  transition matrices.
\newblock Master's thesis, Goethe-University Frankfurt, 2020.

\bibitem[\protect\citeauthoryear{Gordy and Howells}{2004}]{Gordy2004}
Gordy, M.~B. and B.~Howells.
\newblock Procyclicality in {Basel II}: Can we treat the disease without
  killing the patient?
\newblock Working Paper, Board of Governors of the Federal Reserve System,
  April 2004.

\bibitem[\protect\citeauthoryear{G{\"u}ttler and Raupach}{2010}]{Guettler2010}
G{\"u}ttler, A. and P.~Raupach.
\newblock The impact of downward rating momentum.
\newblock {\em Journal of Financial Services Research}, 37(1):1, 2010.

\bibitem[\protect\citeauthoryear{Heitfield}{2005}]{Heitfield2005}
Heitfield, E.~A.
\newblock Dynamics of rating systems.
\newblock In {\em Studies on the Validation of Internal Rating Systems}. BIS
  Working Papers, No. 14, 2005.

\bibitem[\protect\citeauthoryear{Jarrow \bgroup \em et al.\egroup
  }{1997}]{Jarrow1997}
Jarrow, R.~A., D.~Lando, and S.~M. Turnbull.
\newblock A {M}arkov model for the term structure of credit risk spreads.
\newblock {\em Review of Financial Studies}, 10(2):481--523, 1997.

\bibitem[\protect\citeauthoryear{Lando and Sk{\o}deberg}{2002}]{Lando2002}
Lando, D. and T.~M. Sk{\o}deberg.
\newblock Analyzing rating transitions and rating drift with continuous
  observations.
\newblock {\em Journal of Banking \& Finance}, 26(2-3):423--444, 2002.

\bibitem[\protect\citeauthoryear{Lando}{2004}]{Lando2004}
Lando, D.
\newblock {\em Credit Risk Modeling}.
\newblock Princeton University Press, 2004.

\bibitem[\protect\citeauthoryear{Lemmon \bgroup \em et al.\egroup
  }{2008}]{Lemmon2008}
Lemmon, M.~L., M.~R. Roberts, and J.~F. Zender.
\newblock Back to the beginning: persistence and the cross-section of corporate
  capital structure.
\newblock {\em Journal of Finance}, 63(4):1575--1608, 2008.

\bibitem[\protect\citeauthoryear{Malik and Thomas}{2012}]{Malik2012}
Malik, M. and L.~C. Thomas.
\newblock Transition matrix models of consumer credit ratings.
\newblock {\em International Journal of Forecasting}, 28(1):261--272, 2012.

\bibitem[\protect\citeauthoryear{Merton}{1974}]{Merton1974}
Merton, R.~C.
\newblock On the pricing of corporate debt: The risk structure of interest
  rates.
\newblock {\em Journal of Finance}, 29(2):449--470, May 1974.

\bibitem[\protect\citeauthoryear{Miu and Ozdemir}{2017}]{Miu2017}
Miu, P. and B.~Ozdemir.
\newblock Adapting the {B}asel {II} advanced internal-ratings-based models for
  {I}nternational {F}inancial {R}eporting {S}tandard 9.
\newblock {\em Journal of Credit Risk}, 13(2):53--83, 2017.

\bibitem[\protect\citeauthoryear{Morone and Cornaglia}{2009}]{Morone2009}
Morone, M. and A.~Cornaglia.
\newblock Rating philosophy and dynamic properties of internal rating systems:
  a general framework and an application to backtesting.
\newblock {\em Journal of Risk Model Validation}, 3(4):61--88, 2009.

\bibitem[\protect\citeauthoryear{M{\"u}ller and Stoyan}{2002}]{Mueller2002}
M{\"u}ller, A. and D.~Stoyan.
\newblock {\em Comparison methods for stochastic models and risks}, volume 389.
\newblock Wiley, 2002.

\bibitem[\protect\citeauthoryear{M\"uller}{2013}]{Mueller2013}
M\"uller, A.
\newblock {\em Stochastic Orders in Reliability and Risk}, chapter Duality
  Theory and Transfers for Stochastic Order Relations, pages 41--57.
\newblock Springer New York, 2013.

\bibitem[\protect\citeauthoryear{Nickell \bgroup \em et al.\egroup
  }{2000}]{Nickell2000}
Nickell, P., W.~Perraudin, and S.~Varotto.
\newblock Stability of rating transitions.
\newblock {\em Journal of Banking \& Finance}, 24:203--227, 2000.

\bibitem[\protect\citeauthoryear{Norris}{1998}]{Norris1998}
Norris, J.~R.
\newblock {\em Markov Chains}.
\newblock Cambridge University Press, 1998.

\bibitem[\protect\citeauthoryear{Perron}{1907}]{Perron1907}
Perron, O.
\newblock Zur {T}heorie der {M}atrices.
\newblock {\em Mathematische Annalen}, 64(2):248--263, 1907.

\bibitem[\protect\citeauthoryear{Rolski \bgroup \em et al.\egroup
  }{2009}]{Rolski2009}
Rolski, T., H.~Schmidli, V.~Schmidt, and J.~L. Teugels.
\newblock {\em Stochastic processes for insurance and finance}.
\newblock John Wiley \& Sons, 2 edition, 2009.

\bibitem[\protect\citeauthoryear{Rosenblatt}{1971}]{Rosenblatt1971}
Rosenblatt, M.
\newblock {\em Markov processes, Structure and Asymptotic Behavior}.
\newblock Springer, 1971.

\bibitem[\protect\citeauthoryear{Rubtsov and Petrov}{2016}]{Rubtsov2016}
Rubtsov, M. and A.~Petrov.
\newblock A point-in-time-through-the-cycle approach to rating assignment and
  probability of default calibration.
\newblock {\em Journal of Risk Model Validation}, 10(2):83--112, 2016.

\bibitem[\protect\citeauthoryear{Rubtsov}{2021}]{Rubtsov2021}
Rubtsov, M.
\newblock Calibration of rating grades to point-in-time and through-the-cycle
  levels of probability of default.
\newblock {\em Journal of Risk Model Validation}, 2021.

\bibitem[\protect\citeauthoryear{Seneta}{1973}]{Seneta1973}
Seneta, E.
\newblock {\em Non-negative matrices}.
\newblock George Allen \& Unwin Ltd, 1973.

\bibitem[\protect\citeauthoryear{Shaked and Shanthikumar}{2007}]{Shaked2007}
Shaked, M. and J.~G. Shanthikumar.
\newblock {\em Stochastic orders}.
\newblock Springer Science \& Business Media, 2007.

\bibitem[\protect\citeauthoryear{Vall\'es}{2006}]{Valles2006}
Vall\'es, V.
\newblock Stability of a "through-the-cycle" rating system during a financial
  crisis.
\newblock Financial Stability Institute, Bank of International Settlements
  (BIS), 2006.

\bibitem[\protect\citeauthoryear{Wilson}{1998}]{Wilson1998}
Wilson, T.~C.
\newblock Portfolio credit risk.
\newblock Federal Reserve Bank of New York Economic Policy Review, 1998.

\end{thebibliography}
